\documentclass[letterpaper, 11 pt, onecolumn]{ieeeconf}  

\usepackage[utf8]{inputenc}

\usepackage{amscd}
\usepackage{amsfonts}
\usepackage{amsmath}
\usepackage{amssymb}

\usepackage[usenames,dvipsnames]{color}
\usepackage{array}
\usepackage{longtable}
\usepackage{xspace}
\usepackage{paralist}
\usepackage{listings}
\usepackage{extarrows}
\usepackage{xifthen}
\usepackage{url}
  \usepackage{mathpartir}
\usepackage{csquotes}
\usepackage{psfrag}

\newtheorem{theorem}{Theorem}
\newtheorem{example}{Example}
\newtheorem{lemma}{Lemma}
\newtheorem{lemmaRef}[lemma]{Lemma}
\newtheorem{remark}{Remark}
\newtheorem{definition}{Definition}

\newcommand{\smalllb}{\\[-0.25cm]}

\let\exampleOrig\endexample
\def\endexample{\hspace*{0pt}\hfill$\triangleleft$\smalllb\exampleOrig}
\let\lemmaOrig\endlemmaRef
\def\endlemmaRef{\hspace*{0pt}\hfill$\triangleleft$\smalllb\lemmaOrig}
\let\remarkOrig\endremark
\def\endremark{\hspace*{0pt}\hfill$\triangleleft$\smalllb\remarkOrig}
\let\definitionOrig\enddefinition
\def\enddefinition{\hspace*{0pt}\hfill$\triangleleft$\smalllb\definitionOrig}

\usepackage{pgf}
\usepackage{tikz}
\usetikzlibrary{trees,decorations,arrows,automata,shadows,positioning,plotmarks,backgrounds,shapes}
\usetikzlibrary{calc,matrix,fit,petri,decorations.pathmorphing,patterns}

 \pgfdeclarelayer{foreground}
 \pgfdeclarelayer{background}
 \pgfsetlayers{background,main,foreground}

\newcommand{\REFlem}[1]{\text{Lemma~\ref{#1}}}

\newcommand{\REFthm}[1]{\text{Thm.~\ref{#1}}}

\newcommand{\REFdef}[1]{Def.~\ref{#1}}

\newcommand{\REFsec}[1]{Sec.~\ref{#1}}

\newcommand{\deff}{:=}

\newcommand{\BR}[1]{\left( #1 \right)}

\newcommand{\ON}[1]{\operatorname{#1}}

\def\clap#1{\hbox to 0pt{\hss#1\hss}}

\newcommand{\val}[1]{\ensuremath{\mathsf{#1}}}

\newcommand{\DiCases}[4]{\ensuremath{\begin{cases}%
#1&,~#2\\%
#3&,~#4%
\end{cases}}}%
\makeatletter 
\newif\ifFIRST
\newif\ifSECOND
\let\LISTOP\relax
\newcommand{\List}[4][\;]{#3#1%
	\FIRSTtrue
	\@for\i:=#2\do{%
	\ifFIRST\LISTOP{\i}\FIRSTfalse\else,\LISTOP{\i}\fi%
	}%
	#1#4%
	\let\LISTOP\relax
}
\makeatother
\newcounter{DINGLIST}
\stepcounter{DINGLIST}
\newcommand{\markD}[3][\;\;]{\text{\ding{\the\numexpr171+\theDINGLIST}\stepcounter{DINGLIST}}#1#3}

\newcommand{\ZZ}{\textsl{Zu Zeigen:}~\@ifstar\ZZStar\ZZNoStar}
\newcommand{\ZZStar}[1]{\begin{align*}#1\end{align*}}
\newcommand{\ZZNoStar}[1]{\ensuremath{#1}}
\newcommand{\THATIS}{i.e.\xspace}
\newcommand{\SUCHTHAT}{s.t.\xspace}
\newcommand{\IFF}{\unskip~\text{iff}~}
\newcommand{\SHOW}[2][]{Show \ifthenelse{\isempty{#1}}{}{#1, \THATIS, }\ensuremath{#2}:}

\newcommand{\SORRY}[1]{\emph{\color{red}Sorry: #1}\@latex@warning{SORRY: #1}}

\makeatletter

\newcommand{\propNeg}{\@ifstar\propNegStar\propNegNoStar}
\newcommand{\propNegStar}[1]{\ensuremath{\left(\propNegNoStar{#1}\right)}}
\newcommand{\propNegNoStar}[2][\cdot]{\ensuremath{\neg\ifthenelse{\isempty{#2}}{#1}{#2}}}

\newcommand{\propConj}{\@ifstar\propConjStar\propConjNoStar}
\newcommand{\propConjStar}[2]{\ensuremath{\left(\propConjNoStar{#1}{#2}\right)}}
\newcommand{\propConjNoStar}[3][\cdot]{\ensuremath{\ifthenelse{\isempty{#2}}{#1}{#2}\wedge\ifthenelse{\isempty{#3}}{#1}{#3}}}

\newcommand{\propDisj}{\@ifstar\propDisjStar\propDisjNoStar}
\newcommand{\propDisjStar}[2]{\ensuremath{\left(\propDisjNoStar{#1}{#2}\right)}}
\newcommand{\propDisjNoStar}[3][\cdot]{\ensuremath{\ifthenelse{\isempty{#2}}{#1}{#2}\vee\ifthenelse{\isempty{#3}}{#1}{#3}}}

\newcommand{\propImp}{\@ifstar\propImpStar\propImpNoStar}
\newcommand{\propImpStar}[2]{\ensuremath{\left(\propImpNoStar{#1}{#2}\right)}}
\newcommand{\propImpNoStar}[3][\cdot]{\ensuremath{\ifthenelse{\isempty{#2}}{#1}{#2}\Rightarrow\ifthenelse{\isempty{#3}}{#1}{#3}}}

\newcommand{\propAequ}{\@ifstar\propAequStar\propAequNoStar}
\newcommand{\propAequStar}[2]{\ensuremath{\left(\propAequNoStar{#1}{#2}\right)}}
\newcommand{\propAequNoStar}[3][\cdot]{\ensuremath{\ifthenelse{\isempty{#2}}{#1}{#2}\Leftrightarrow\ifthenelse{\isempty{#3}}{#1}{#3}}}

\newcommand{\propXOR}{\@ifstar\propXORStar\propXORNoStar}
\newcommand{\propXORStar}[2]{\ensuremath{\left(\propXORNoStar{#1}{#2}\right)}}
\newcommand{\propXORNoStar}[3][\cdot]{\ensuremath{\ifthenelse{\isempty{#2}}{#1}{#2}\oplus\ifthenelse{\isempty{#3}}{#1}{#3}}}

\newcommand{\AllQ}{\@ifstar\AllQStar\AllQNoStar}
\newcommand{\AllQStar}[3][\;]{\ensuremath{\left(\forall #2#1.#1#3\right)}}
\newcommand{\AllQNoStar}[3][\;]{\ensuremath{\forall #2#1.#1#3}}
\newcommand{\AllQu}{\@ifstar\AllQuStar\AllQuNoStar}
\newcommand{\AllQuStar}[3][\;]{\ensuremath{\left(\forall^{\infty} #2#1.#1#3\right)}}
\newcommand{\AllQuNoStar}[3][\;]{\ensuremath{\forall^{\infty} #2#1.#1#3}}

\newcommand{\ExQ}{\@ifstar\ExQStar\ExQNoStar}
\newcommand{\ExQStar}[3][\;]{\ensuremath{\left(\exists #2#1.#1#3\right)}}
\newcommand{\ExQNoStar}[3][\;]{\ensuremath{\exists #2#1.#1#3}}

\newcommand{\NExQ}{\@ifstar\NExQStar\NExQNoStar}
\newcommand{\NExQStar}[3][\;]{\ensuremath{\left(\nexists #2#1.#1#3\right)}}
\newcommand{\NExQNoStar}[3][\;]{\ensuremath{\nexists #2#1.#1#3}}

\newcommand{\UniqueQ}{\@ifstar\UniqueQStar\UniqueQNoStar}
\newcommand{\UniqueQStar}[3][\;]{\ensuremath{\left(\exists! #2#1.#1#3\right)}}
\newcommand{\UniqueQNoStar}[3][\;]{\ensuremath{\exists! #2#1.#1#3}}

\newcommand{\Set}[2][]{\List[#1]{#2}{\{}{\}}}
\newcommand{\VSet}[2][]{\let\LISTOP\val\List[#1]{#2}{\{}{\}}}

\newcommand{\Tuple}[2][]{\List[#1]{#2}{(}{)}}
\newcommand{\VTuple}[2][]{\let\LISTOP\val\List[#1]{#2}{(}{)}}

\newcommand{\SetComp}[3][]{\{#1#2#1\mid#1#3#1\}}
\newcommand{\SetCompX}[3][]{\left\{#1#2#1\middle\vert#1#3#1\right\}}

\newcommand{\POWERSET}{\@ifstar\POWERSETStar\POWERSETNoStar}
\newcommand{\POWERSETStar}[1]{\ensuremath{\ON{2}^{\ifthenelse{\isempty{#1}}{\cdot}{#1}}}}
\newcommand{\POWERSETNoStar}[1]{\ensuremath{\ON{2}^{\ifthenelse{\isempty{#1}}{\cdot}{#1}}}}

\newcommand{\FINPOWERSET}{\@ifstar\FINPOWERSETStar\FINPOWERSETNoStar}
\newcommand{\FINPOWERSETStar}[1]{\ensuremath{\mathcal{P}_{\ON{fin}}(\ifthenelse{\isempty{#1}}{\cdot}{#1})}}
\newcommand{\FINPOWERSETNoStar}[1]{\ensuremath{\mathcal{P}_{\ON{fin}}\left(\ifthenelse{\isempty{#1}}{\cdot}{#1}\right)}}

\newcommand{\UNION}{\@ifstar\UNIONStar\UNIONNoStar}
\newcommand{\UNIONStar}[2]{\ensuremath{\left(\UNIONNoStar{#1}{#2}\right)}}
\newcommand{\UNIONNoStar}[2]{\ensuremath{\ifthenelse{\isempty{#1}}{\cdot}{#1}\cup\ifthenelse{\isempty{#2}}{\cdot}{#2}}}

\newcommand{\UNIOND}{\@ifstar\UNIONDStar\UNIONDNoStar}
\newcommand{\UNIONDStar}[2]{\ensuremath{\left(\UNIONDNoStar{#1}{#2}\right)}}
\newcommand{\UNIONDNoStar}[2]{\ensuremath{\ifthenelse{\isempty{#1}}{\cdot}{#1}\uplus\ifthenelse{\isempty{#2}}{\cdot}{#2}}}

\newcommand{\SETMINUS}{\@ifstar\SETMINUSStar\SETMINUSNoStar}
\newcommand{\SETMINUSStar}[2]{\ensuremath{\left(\SETMINUSNoStar{#1}{#2}\right)}}
\newcommand{\SETMINUSNoStar}[2]{\ensuremath{\ifthenelse{\isempty{#1}}{\cdot}{#1}\setminus\ifthenelse{\isempty{#2}}{\cdot}{#2}}}

\newcommand{\INTERSECT}{\@ifstar\INTERSECTStar\INTERSECTNoStar}
\newcommand{\INTERSECTStar}[2]{\ensuremath{\left(\INTERSECTNoStar{#1}{#2}\right)}}
\newcommand{\INTERSECTNoStar}[2]{\ensuremath{\ifthenelse{\isempty{#1}}{\cdot}{#1}\cap\ifthenelse{\isempty{#2}}{\cdot}{#2}}}

\newcommand{\CARTPROD}{\@ifstar\CARTPRODStar\CARTPRODNoStar}
\newcommand{\CARTPRODStar}[2]{\ensuremath{\left(\CARTPRODNoStar{#1}{#2}\right)}}
\newcommand{\CARTPRODNoStar}[2]{\ensuremath{\ifthenelse{\isempty{#1}}{\cdot}{#1}\times\ifthenelse{\isempty{#2}}{\cdot}{#2}}}

\newcommand{\FINCOUNT}{\@ifstar\FinCountStar\FinCountNoStar}
\newcommand{\FinCountStar}[1]{\ensuremath{\#(\ifthenelse{\isempty{#1}}{\cdot}{#1})}}
\newcommand{\FinCountNoStar}[1]{\ensuremath{\#\left(\ifthenelse{\isempty{#1}}{\cdot}{#1}\right)}}

\makeatother 

\newcommand{\sconc}{\cdot}

\newcommand{\parfun}{\ensuremath{\ON{\rightharpoonup}}}
\newcommand{\fun}{\ensuremath{\ON{\rightarrow}}}

\tikzstyle{istate}=[state,initial,initial text=]
\tikzstyle{fistate}=[state,accepting,initial,initial text=]
\tikzstyle{fistateA}=[state,accepting,initial,initial text=,initial where=above]
\tikzstyle{fistateB}=[state,accepting,initial,initial text=,initial where=below]
\tikzstyle{fistateL}=[state,accepting,initial,initial text=,initial where=left]
\tikzstyle{fistateR}=[state,accepting,initial,initial text=,initial where=right]
\tikzstyle{ifstate}=[state,accepting,initial,initial text=]
\tikzstyle{ifstateA}=[state,accepting,initial,initial text=,initial where=above]
\tikzstyle{ifstateB}=[state,accepting,initial,initial text=,initial where=below]
\tikzstyle{ifstateL}=[state,accepting,initial,initial text=,initial where=left]
\tikzstyle{ifstateR}=[state,accepting,initial,initial text=,initial where=right]
\tikzstyle{istateA}=[state,initial,initial text=,initial where=above]
\tikzstyle{istateB}=[state,initial,initial text=,initial where=below]
\tikzstyle{istateL}=[state,initial,initial text=,initial where=left]
\tikzstyle{istateR}=[state,initial,initial text=,initial where=right]
\tikzstyle{fstate}=[state,accepting]
\tikzstyle{SFSautomat}=[->,>=stealth',shorten >=1pt,auto,node distance=2cm,on grid,semithick,inner sep=2pt,bend angle=45]
\newcommand{\SFSAutomatEdge}[5]{\draw[->, thick](#1) edge[#4] node[#5] {\ensuremath{#2}} (#3);}

\newenvironment{propConjA}{\left(\def\unionAtest{1}\begin{array}{@{\if\unionAtest1\gdef\unionAtest{0}\phantom{\wedge}\else\wedge\fi}l@{}}}{\end{array}\right)}

\makeatletter
  \newlength{\SFS@HEIGHT}
  \newlength{\SFS@WIDTH}
  \newcommand{\SplitX}[2]{
	    \settoheight{\SFS@HEIGHT}{$#2$}
	    \settowidth{\SFS@WIDTH}{$#2$}
	    \mbox{\begin{tikzpicture}[baseline=(current bounding box.center)]
	    \node[] (E) at (0,0) {$#1$};
	    \node[inner sep=0pt] (F) at ($(E.south west)+(1ex,-1ex)+(3ex+.5\SFS@WIDTH,-\SFS@HEIGHT)$) {$#2$};
	    \node[] (E) at (0,0) {\phantom{$#1$}};
	    \draw[fill] ($(E.east)+(1ex,0ex)$) circle (.2ex);
	    \draw[-] ($(E.east)+(1ex,0ex)$) -- ($(E.south east)+(1ex,-0.5ex)$) -- ($(E.south west)+(1ex,-0.5ex)$) -- ($(E.south west)+(1ex,-1ex)-(0,\SFS@HEIGHT)$) -- ($(E.south west)+(2.5ex,-1ex)-(0,\SFS@HEIGHT)$);
	    \draw[fill] ($(E.south west)+(2.5ex,-1ex)-(0,\SFS@HEIGHT)$) circle (.2ex);
	    \end{tikzpicture}}}
  \newcommand{\SplitS}[2]{
	    \settoheight{\SFS@HEIGHT}{$#2$}
	    \settowidth{\SFS@WIDTH}{$#2$}
	    \mbox{\begin{tikzpicture}[baseline=(current bounding box.center)]
	    \node[] (E) at (0,0) {$#1$};
	    \node[inner sep=0pt] (F) at ($(E.south west)+(1ex,0.5ex)+(3ex+.5\SFS@WIDTH,-\SFS@HEIGHT)$) {$#2$};
	    \end{tikzpicture}}}	    
  \newcommand{\AllQSplit}[2]{\SplitX{\forall\;#1\;.}{#2}}

  \newcommand{\ExQSplit}[2]{\SplitX{\exists\;#1\;.}{#2}}

\makeatother

\providecommand{\length}[1]{\lvert#1\rvert}
\providecommand{\lengthw}[1]{\lvert#1\rvert_{\hspace{-0.2mm}_L}}

\newcommand{\trivialN}[1]{\text{trivial}\xspace}

\newcommand{\inps}{\ensuremath{\hspace{-0.8mm}\in\hspace{-0.8mm}}}

\newcommand{\gps}{\ensuremath{\hspace{-0.8mm}>\hspace{-0.8mm}}}
\newcommand{\timesps}{\ensuremath{\hspace{-0.8mm}\times\hspace{-0.8mm}}}
\newcommand{\equps}{\ensuremath{\hspace{-0.8mm}=\hspace{-0.8mm}}}

\newcommand{\Rb}{\ensuremath{\mathbb{R}}} 
 
\newcommand{\Nbn}{\ensuremath{\mathbb{N}_{0}}}
\newcommand{\Rbn}{\ensuremath{\mathbb{R}_{0}^+}}

\newcommand{\twoup}[1]{\ensuremath{2^{#1}}} 

\newcommand{\Psys}{\ensuremath{P}} 
\newcommand{\X}{\ensuremath{X}} 
\newcommand{\Xt}[2]{\ensuremath{\ifthenelse{\isempty{#2}}{X_{I}^{#1}}{X_{I,#2}^{#1}}}} 
\newcommand{\XT}[1]{\ensuremath{\ifthenelse{\isempty{#1}}{X_I}{X_{I,#1}}}} 
\newcommand{\Xk}[2]{\ensuremath{\ifthenelse{\isempty{#2}}{X_{E}^{#1}}{X_{E,#2}^{#1}}}} 
\newcommand{\XK}[1]{\ensuremath{\ifthenelse{\isempty{#1}}{X_{E}}{X_{E,#1}}}} 
\newcommand{\Zt}[2]{\ensuremath{\ifthenelse{\isempty{#2}}{Z_{I}^{#1}}{Z_{I,#2}^{#1}}}} 
\newcommand{\ZT}[1]{\ensuremath{\ifthenelse{\isempty{#1}}{Z_I}{Z_{I,#1}}}} 
\newcommand{\Zk}[2]{\ensuremath{\ifthenelse{\isempty{#2}}{Z_{E}^{#1}}{Z_{E,#2}^{#1}}}} 
\newcommand{\ZK}[1]{\ensuremath{\ifthenelse{\isempty{#1}}{Z_{E}}{Z_{E,#1}}}} 

\newcommand{\x}{\ensuremath{x}} 
\newcommand{\xt}{\ensuremath{\tilde{x}}} 
\newcommand{\kt}{\ensuremath{\tilde{k}}} 
\newcommand{\Xo}[1]{\ensuremath{X_{#1 0}}}

\newcommand{\Z}{\ensuremath{Z}}

\newcommand{\ZPit}[1]{\ensuremath{\ifthenelse{\isempty{#1}}{Z_t}{Z_{#1,t}}}} 
\newcommand{\ZPiT}[1]{\ensuremath{\ifthenelse{\isempty{#1}}{Z_T}{Z_{#1,T}}}} 
 
\newcommand{\ZTit}[1]{\ensuremath{\ifthenelse{\isempty{#1}}{\breve{Z}_t}{\breve{Z}_{#1,t}}}} 
\newcommand{\ZTiT}[1]{\ensuremath{\ifthenelse{\isempty{#1}}{\breve{Z}_T}{\breve{Z}_{#1,T}}}} 
\newcommand{\z}{\ensuremath{z}}

\newcommand{\w}{\ensuremath{w}} 
\newcommand{\wt}{\ensuremath{\tilde{w}}} 
\newcommand{\gt}{\ensuremath{\tilde{\gamma}}}
\newcommand{\tet}{\ensuremath{\tilde{t}}}
\newcommand{\taut}{\ensuremath{\tilde{\tau}}}

\newcommand{\tr}{\ensuremath{\delta}}

\newcommand{\T}{\ensuremath{T}} 
\newcommand{\I}{\ensuremath{\mathbf{i}}} 
\newcommand{\Interval}{\ensuremath{\mathcal{I}}}

\renewcommand{\ll}[1]{\ensuremath{|_{[#1]}}} 
\newcommand{\lb}[1]{\ensuremath{|_{\langle#1\rangle}}} 
\newcommand{\llr}[1]{\ensuremath{|_{[#1)}}}

\newcommand{\f}[1]{\ensuremath{f\ifthenelse{\isempty{#1}}{}{\Tuple{#1}}}} 
\newcommand{\g}[1]{\ensuremath{g\ifthenelse{\isempty{#1}}{}{\Tuple{#1}}}} 
\newcommand{\h}[1]{\ensuremath{h\ifthenelse{\isempty{#1}}{}{\Tuple{#1}}}} 
\newcommand{\fs}[2]{\ensuremath{f_{#2}\ifthenelse{\isempty{#1}}{}{\Tuple{#1}}}} 
\newcommand{\gs}[2]{\ensuremath{g_{#2}\ifthenelse{\isempty{#1}}{}{\Tuple{#1}}}}  
\newcommand{\hs}[2]{\ensuremath{h_{#2}\ifthenelse{\isempty{#1}}{}{\Tuple{#1}}}}  

\newcommand{\lnc}{\ensuremath{l_t}}

\newcommand{\Beh}{\ensuremath{\mathcal{B}}}

\newcommand{\kg}[1]{\ensuremath{\xspace\preceq_{#1}\xspace}}
\newcommand{\hg}[1]{\ensuremath{\xspace\cong_{#1}\xspace}}
\newcommand{\async}{\ensuremath{\wr_{|}}}
\newcommand{\sync}{\ensuremath{\shortparallel}}
\newcommand{\wsync}{\ensuremath{\wr\shortmid}}
\newcommand{\SR}[3]{\ensuremath{\mathfrak{R}_{#1}(#2,#3)}}

\newcommand{\ESn}[1]{\ensuremath{\ifthenelse{\isempty{#1}}{\Sigma^+_{S}}{\Sigma^+_{S,#1}}}}

\newcommand{\BehS}[1]{\ensuremath{\ifthenelse{\isempty{#1}}{\Beh_{S}}{\Beh_{S,#1}}}}
\newcommand{\BehE}[1]{\ensuremath{\ifthenelse{\isempty{#1}}{\Beh_{E}}{\Beh_{E,#1}}}}

\newcommand{\WT}{\ensuremath{W}}
\newcommand{\W}{\ensuremath{W}}
\newcommand{\D}{\ensuremath{D}}
\newcommand{\V}{\ensuremath{V}}
\newcommand{\statemap}[3]{\ifthenelse{\isempty{#2#3}}{\psi_{#1}}{\psi_{#1}(#2,#3)}}
\newcommand{\statemapPi}[3]{\ifthenelse{\isempty{#2#3}}{\varphi_{#1}}{\varphi_{#1}(#2,#3)}}
\newcommand{\statemapTi}[3]{\ifthenelse{\isempty{#2#3}}{\psi_{#1}}{\psi_{#1}(#2,#3)}}

\newcommand{\CONCAT}[4]{#1\wedge^{#2}_{#3}#4}

\newcommand{\Xx}[2]{\ensuremath{\ifthenelse{\isempty{#1}}{\mathcal{X}_{E}}{\mathcal{X}_{E,#1}}\ifthenelse{\isempty{#2}}{}{(#2)}}}
\newcommand{\Xxp}[2]{\ensuremath{\overline{\chi}_{#1}\ifthenelse{\isempty{#2}}{}{(#2)}}}
\newcommand{\Xxr}[3]{\ensuremath{\ifthenelse{\isempty{#1}}{\mathcal{X}_{E}^{#2}}{\mathcal{X}_{E,#1}^{#2}}\ifthenelse{\isempty{#3}}{}{(#3)}}}
\newcommand{\Xxrp}[3]{\ensuremath{\overline{\chi}_{#1}^{#2}\ifthenelse{\isempty{#3}}{}{(#3)}}}

\newcommand{\R}{\ensuremath{\mathcal{R}}}

\newcommand{\timescale}{\mathcal{T}}

\newcommand{\timescaleUp}[1]{{#1}}
\newcommand{\timescaleDown}[1]{{#1}^{-1}}
\newcommand{\signalmap}{\phi}
\newcommand{\signalmapR}[1]{\ensuremath{\ifthenelse{\isempty{#1}}{\signalmap_{r}}{\signalmap_{r,#1}}}}
\newcommand{\signalmapP}[1]{\ensuremath{\ifthenelse{\isempty{#1}}{\signalmap_{p}}{\signalmap_{p,#1}}}}
\newcommand{\signalmapT}[1]{\ensuremath{\ifthenelse{\isempty{#1}}{\signalmap_{t}}{\signalmap_{t,#1}}}}

\newcommand{\philaMax}{\signalmap_l}
\newcommand{\E}{\ensuremath{\Sigma}}
\newcommand{\Ep}[1]{\ensuremath{\Sigma_{#1}^{\signalmap}}}

\newcommand{\EpRhsDisc}{\ensuremath{\Tuple{\T,\Nbn,\WT,\Gamma,\Beh,\BehE{},\signalmap}}}
\newcommand{\ES}[1]{\ensuremath{\ifthenelse{\isempty{#1}}{\Sigma_{S}}{\Sigma_{S,#1}}}}
\newcommand{\EpS}[1]{\ensuremath{\ifthenelse{\isempty{#1}}{\Ep{S}}{\Ep{S,#1}}}}
\newcommand{\EpSR}[1]{\ensuremath{\ifthenelse{\isempty{#1}}{\Sigma^{\signalmap_{r}}_{S}}{\Sigma^{\signalmap_{r,#1}}_{S,#1}}}}
\newcommand{\EpSP}[1]{\ensuremath{\ifthenelse{\isempty{#1}}{\Sigma^{\signalmap_{p}}_{S}}{\Sigma^{\signalmap_{p,#1}}_{S,#1}}}}
\newcommand{\EpST}[1]{\ensuremath{\ifthenelse{\isempty{#1}}{\Sigma^{\signalmap_{t}}_{S}}{\Sigma^{\signalmap_{t,#1}}_{S,#1}}}}

\newcommand{\EpSRhsDisc}[1]{\ensuremath{(\T_{#1},\allowbreak\Nbn,\allowbreak\WT_{#1}\nobreak\times\nobreak\X_{#1},\allowbreak\Gamma,\allowbreak\BehS{#1},\allowbreak\BehE{#1},\allowbreak\signalmap_{#1})}}
\newcommand{\EplaMaxSRhs}[1]{\ensuremath{(\Nbn,\allowbreak\Nbn,\allowbreak\Gamma_{#1}\nobreak\times\nobreak\Z_{#1},\allowbreak\Gamma,\allowbreak\BehlaMaxS{#1},\allowbreak\BehElaMax{#1},\allowbreak\philaMax)}}
\newcommand{\EE}[1]{\ensuremath{\ifthenelse{\isempty{#1}}{\Sigma_{E}}{\Sigma_{E,#1}}}}
\newcommand{\ESm}[1]{\ensuremath{\ifthenelse{\isempty{#1}}{\Sigma_{\psi}}{\Sigma_{\psi,#1}}}}

\newcommand{\EElaMax}{\ensuremath{\Sigma_{E}^{l^\uparrow}}}

\newcommand{\EplMaxS}[1]{\ensuremath{\ifthenelse{\isempty{#1}}{\Sigma^{\signalmap,l^\Uparrow}_{S}}{\Sigma^{\signalmap,l^\Uparrow}_{S,#1}}}}
\newcommand{\EplaMaxS}[1]{\ensuremath{\ifthenelse{\isempty{#1}}{\Sigma^{\signalmap,l^\uparrow}_{S}}{\Sigma^{\signalmap,l^\uparrow}_{S,#1}}}}
\newcommand{\EplMaxSp}[1]{\ensuremath{\ifthenelse{\isempty{#1}}{\overline{\Sigma}^{\signalmap,l^\uparrow}_{S}}{\overline{\Sigma}^{\signalmap,l^\uparrow}_{S,#1}}}}
\newcommand{\EplncMaxS}[1]{\ensuremath{\ifthenelse{\isempty{#1}}{\Sigma^{\signalmap,\lnc^\uparrow}_{S}}{\Sigma^{\signalmap,\lnc^\uparrow}_{S,#1}}}}
\newcommand{\EplsMaxS}[2]{\ensuremath{\ifthenelse{\isempty{#1}}{\Sigma^{\signalmap,{#2}^\uparrow}_{S}}{\Sigma^{\signalmap,{#2}^\uparrow}_{S,#1}}}}

\newcommand{\ElMaxS}[1]{\ensuremath{\ifthenelse{\isempty{#1}}{\Sigma^{l^\Uparrow}_{S}}{\Sigma^{l^\Uparrow}_{S,#1}}}}
\newcommand{\ElaMaxS}[1]{\ensuremath{\ifthenelse{\isempty{#1}}{\Sigma^{l^\uparrow}_{S}}{\Sigma^{l^\uparrow}_{S,#1}}}}
\newcommand{\ElMaxSp}[1]{\ensuremath{\ifthenelse{\isempty{#1}}{\overline{\Sigma}^{l^\uparrow}_{S}}{\overline{\Sigma}^{l^\uparrow}_{S,#1}}}}
\newcommand{\ElncMaxS}[1]{\ensuremath{\ifthenelse{\isempty{#1}}{\Sigma^{\lnc^\uparrow}_{S}}{\Sigma^{\lnc^\uparrow}_{S,#1}}}}
\newcommand{\ElsMaxS}[2]{\ensuremath{\ifthenelse{\isempty{#1}}{\Sigma^{{#2}^\uparrow}_{S}}{\Sigma^{{#2}^\uparrow}_{S,#1}}}}

\newcommand{\ElE}[1]{\ensuremath{\ifthenelse{\isempty{#1}}{\Sigma^{l}_{E}}{\Sigma^{l}_{E,#1}}}}
\newcommand{\ElEp}[1]{\ensuremath{\ifthenelse{\isempty{#1}}{\Sigma^{l}_{E}}{\overline{\Sigma}^{l}_{E,#1}}}}
\newcommand{\ElncE}[1]{\ensuremath{\ifthenelse{\isempty{#1}}{\Sigma^{\lnc}_{E}}{\Sigma^{\lnc}_{E,#1}}}}
\newcommand{\ElMaxE}[1]{\ensuremath{\ifthenelse{\isempty{#1}}{\Sigma^{l^\uparrow}_{E}}{\Sigma^{l^\uparrow}_{E,#1}}}}
\newcommand{\ElMaxEp}[1]{\ensuremath{\ifthenelse{\isempty{#1}}{\overline{\Sigma}^{l^\uparrow}_{E}}{\overline{\Sigma}^{l^\uparrow}_{E,#1}}}}
\newcommand{\ElncMaxE}[1]{\ensuremath{\ifthenelse{\isempty{#1}}{\Sigma^{\lnc^\uparrow}_{E}}{\Sigma^{\lnc^\uparrow}_{E,#1}}}}
\newcommand{\ElsMaxE}[2]{\ensuremath{\ifthenelse{\isempty{#1}}{\Sigma^{{#2}^\uparrow}_{E}}{\Sigma^{{#2}^\uparrow}_{E,#1}}}}

\newcommand{\BehElaMax}{\ensuremath{\Beh_{E}^{l^\uparrow}}}

\newcommand{\BehlMaxS}[1]{\ensuremath{\ifthenelse{\isempty{#1}}{\Beh^{l^\Uparrow}_{S}}{\Beh^{l^\Uparrow}_{S,#1}}}}
\newcommand{\BehlaMaxS}[1]{\ensuremath{\ifthenelse{\isempty{#1}}{\Beh^{l^\uparrow}_{S}}{\Beh^{l^\uparrow}_{S,#1}}}}
\newcommand{\BehlMaxSp}[1]{\ensuremath{\ifthenelse{\isempty{#1}}{\overline{\Beh}^{l^\uparrow}_{S}}{\overline{\Beh}^{l^\uparrow}_{S,#1}}}}
\newcommand{\BehlncMaxS}[1]{\ensuremath{\ifthenelse{\isempty{#1}}{\Beh^{\lnc^\uparrow}_{S}}{\Beh^{\lnc^\uparrow}_{S,#1}}}}
\newcommand{\BehlsMaxS}[2]{\ensuremath{\ifthenelse{\isempty{#1}}{\Beh^{{#2}^\uparrow}_{S}}{\Beh^{{#2}^\uparrow}_{S,#1}}}}

\newcommand{\BehlE}[1]{\ensuremath{\ifthenelse{\isempty{#1}}{\Beh^{l}_{E}}{\Beh^{l}_{E,#1}}}}
\newcommand{\BehlncE}[1]{\ensuremath{\ifthenelse{\isempty{#1}}{\Beh^{\lnc}_{E}}{\Beh^{\lnc}_{E,#1}}}}
\newcommand{\BehlMaxE}[1]{\ensuremath{\ifthenelse{\isempty{#1}}{\Beh^{l^\uparrow}_{E}}{\Beh^{l^\uparrow}_{E,#1}}}}
\newcommand{\BehlMaxEp}[1]{\ensuremath{\ifthenelse{\isempty{#1}}{\overline{\Beh}^{l^\uparrow}_{E}}{\overline{\Beh}^{l^\uparrow}_{E,#1}}}}
\newcommand{\BehlncMaxE}[1]{\ensuremath{\ifthenelse{\isempty{#1}}{\Beh^{\lnc^\uparrow}_{E}}{\Beh^{\lnc^\uparrow}_{E,#1}}}}
\newcommand{\BehlsMaxE}[2]{\ensuremath{\ifthenelse{\isempty{#1}}{\Beh^{{#2}^\uparrow}_{E}}{\Beh^{{#2}^\uparrow}_{E,#1}}}}

\newcommand{\BehVlMaxS}[1]{\ensuremath{\ifthenelse{\isempty{#1}}{\projState{\V}{\Beh^{l^\uparrow}_{S}}}{\projState{\V}{\Beh^{l^\uparrow}_{S,#1}}}}}

\newcommand{\BehDlMaxS}[1]{\ensuremath{\ifthenelse{\isempty{#1}}{\projState{\D}{\Beh^{l^\uparrow}_{S}}}{\projState{\D}{\Beh^{l^\uparrow}_{S,#1}}}}}

\newcommand{\EoMaxS}[1]{\ensuremath{\ifthenelse{\isempty{#1}}{\Sigma^{1^\uparrow}_{S}}{\Sigma^{1^\uparrow}_{S,#1}}}}

\newcommand{\projState}[2]{\pi_{#1}(#2)}

\newcommand{\dom}[1]{\ensuremath{\mathrm{dom(#1)}}}

\title{\LARGE \bf
Constructing (Bi)Similar Finite State Abstractions using Asynchronous $l$-Complete Approximations
}

\author{Anne-Kathrin Schmuck and Jörg Raisch
\thanks{A.-K. Schmuck and J. Raisch are with the Control Systems Group, Technical University of Berlin, Germany. J. Raisch is also with the Max Planck Institute
for Dynamics of Complex Technical Systems, Magdeburg, Germany. {\tt\small \{a.schmuck,raisch\}@control.tu-berlin.de}}
}

\begin{document}

\maketitle
\thispagestyle{empty}
\pagestyle{empty}

 \begin{abstract}
This paper constructs a finite state abstraction of a possibly continuous-time and infinite state model in two steps. First, a finite external signal space is added, generating a so called $\signalmap$-dynamical system. Secondly, the strongest asynchronous $l$-complete approximation of the external dynamics is constructed. As our main results, we show that 
\begin{inparaenum}[(i)]
 \item the abstraction simulates the original system, and 
 \item bisimilarity between the original system and its abstraction holds, if and only if the original system is $l$-complete and its state space satisfies an additional property.
\end{inparaenum}
\end{abstract}

\section{Introduction}
Real life control problems for large scale systems are often very challenging due to numerous interactions between different components and usually tight performance requirements. One way to reduce the complexity of such problems is to introduce different control layers using well defined abstractions. Usually, the top control layer will enforce high level specifications, such as interconnection or safety requirements, typically expressible by regular languages. 
With this specification type, supervisory control theory (SCT) \cite{RamWon1984} can be used to synthesize a correct by design controller if the abstracted model can be represented by a regular language, i.e., if it can be realized by a finite state machine.\\
Motivated by this, Tabuada and Pappas \cite{TabuadaPappas2003b,TabuadaPappas2003,TabuadaBook} developed finite state abstraction methods 
generating a regular language representation of the plant model.
With the same motivation but independently from their work, the notion of a strongest $l$-complete approximation was introduced by Moor and Raisch \cite{MoorRaisch1999, MoorRaischYoung2002}  as a discrete abstraction technique for time invariant behavioral systems. The applicability of this approximation method was recently enlarged by Schmuck and Raisch \cite{SchmuckRaisch2014_ControlLetters}, introducing so called asynchronous $l$-complete approximations.\\
While the existence of simulation or bisimulation relations between $l$-complete approximations and the original system has not yet been formally investigated, the abstraction techniques by Tabuada and Pappas ensure the existence of such relations between the original and the abstracted plant model. However, in their work, 
the original system is rewritten into a transition system, previous to the abstraction step. The simulation or bisimulation relation is then ensured to hold between the transition system (not the original model) and its finite state abstraction.
Their rewriting step is necessary since simulation relations are naturally defined between models evolving on the same time axis. To overcome this limitation, Schmuck and Raisch \cite{SchmuckRaisch2014_HSCC} introduced $\signalmap$-dynamical systems, a system model with distinct external and internal signals possibly evolving on different time axes. In \cite{SchmuckRaisch2014_HSCC}, different notions of simulation and bisimulation where derived, ensuring that they are, respectively, preorders and equivalence relations for this system class.\\
$\signalmap$-dynamical systems are able to model abstraction processes or signal aggregation by combining both the original (possibly continuous-time) state dynamics and the corresponding external discrete-time behavior. This can naturally be combined with asynchronous $l$-complete approximations of the external behavior, generating a finite state abstraction if the external signal space is finite. Therefore, in contrast to the work by Tabuada and Pappas, no intermediate transition system has to be introduced to reason about similarity.\\
After introducing required notation in Section~\ref{sec:prelim}, we review the notion of $\signalmap$-dynamical systems in Section~\ref{sec:sys} and apply the construction of a strongest asynchronous $l$-complete approximation to this system class in Section~\ref{sec:lcomp}. In Section~\ref{sec:SimRel} we briefly review the simulation relations defined in \cite{SchmuckRaisch2014_HSCC} for $\signalmap$-dynamical systems. As our main result, we prove the existence of different simulation relations between the original system and its approximation and derive necessary and sufficient conditions for bisimilarity in Section~\ref{sec:SimRelApprox}.\\
Our construction extends the work by Tabuada and Pappas in three ways:
\begin{inparaenum}[(i)]
 \item simulation relations are established between the original state space dynamics and the abstraction,
 \item the accuracy of the abstracted system can be adjusted during construction without refining the external signal space and
 \item the behavioral framework (e.g., \cite{Willems1991}) is used%
 , allowing for infinite trajectories with eventuality properties.
\end{inparaenum}

\section{Preliminaries}\label{sec:prelim}

In the behavioral framework (e.g., \cite{Willems1991}), a \textit{dynamical system} is given by ${\E=\Tuple{\T,\WT,\Beh}}$, consisting of the right-unbounded time axis ${\T\subseteq\Rb}$, the signal space $\WT$ and the behavior of the system, ${\Beh\subseteq\WT^\T}$, where
$\WT^\T\deff\SetComp{w}{w:\T\fun\WT}$ is the set of all 
\textit{signals} evolving on $\T$ and
taking values in $\WT$. 
Slightly abusing notation, we also write  $v\in\WT^\T$ if $v:\T\parfun\WT$ is a \textit{partial function}. This is understood to be shorthand for $v\in\WT^{\dom{v}}$, where $\dom{v}=\SetComp{t\in\T}{v(t)\text{ is defined}}$ is the \textit{domain} of $v$. 
Furthermore, $\I:\T\fun \T$ is the \textit{identity map} s.t.
\footnote{Throughout this paper we use the notation "$\AllQ{}{}$", meaning that all statements after the dot hold for all variables in front of the dot. "$\ExQ{}{}$" is interpreted analogously. } 
$\AllQ{t\in\T}{\I(t)=t}$.\\
Let $\Interval$ be a bounded interval on $\T$, then $\WT^\Interval\deff\SetComp{w}{w:\Interval\fun\WT}$ is the set of \textit{signals on $\Interval$} taking values in $\WT$. 
 Furthermore, $\w|_{\Interval}$ is the \textit{restriction} of the map $w:\T\fun\WT$ to the domain $\Interval$.
 $\Beh|_{\Interval}\subseteq\WT^\Interval$ denotes the restriction of all trajectories in $\Beh$ to $\Interval$. 
Now let $\WT=\WT_1\times\WT_2$ be a product space. Then the \textit{projection} of a signal $w\in\WT^\T$ to $\WT_1$ is given by ${\projState{\WT_1}{w}\deff\SetComp{w_1\in\WT_1^\T}{\ExQ{w_2\in\WT_2^\T}{w=\Tuple{w_1,w_2}}}}$ and $\projState{\WT_1}{\Beh}$ denotes the projection of all signals in the behavior to $\WT_1$.  
Given two signals $w_1,w_2\in\WT^\T$ and two time points $t_1,t_2\in\T$, the \textit{concatenation}  $w_3=\CONCAT{w_1}{t_1}{t_2}{w_2}$ is given by
\begin{equation}\label{equ:concat}
\AllQ{t\in\T}{w_3(t)=\DiCases{w_1(t)}{t< t_1}{w_2(t-t_1+t_2)}{t\geq t_1},}
\end{equation}
where we denote $\CONCAT{\cdot}{t}{t}{\cdot}$ by $\CONCAT{\cdot}{}{t}{\cdot}$.\\
Following \cite{Willems1991}, a system ${\E=\Tuple{\T,\WT,\Beh}}$ is \textit{complete} if
\begin{equation}\label{equ:complete_old}
\propAequ{\AllQ*{t_1,t_2\in\T,t_1\leq t_2}{w\ll{t_1,t_2}\in \Beh\ll{t_1,t_2}}}{w\in\Beh},
\end{equation}
and, following \cite[Def.3]{SchmuckRaisch2014_ControlLetters}, we say that ${\E=\Tuple{\T,\WT,\Beh}}$ is \textit{asynchronously $l$-complete} if
\begin{equation}\label{equ:lcomplete}
\propAequ{
\begin{propConjA}
w\ll{0,l}\in\Beh\ll{0,l}\\
 \AllQ{t\in\T}{w\ll{t,t+l}\in \bigcup_{t'\in\T}\Beh\ll{t',t'+l}}
\end{propConjA}}{w\in\Beh}.
\end{equation}
Now let $X$ be a set. Then, following \cite[Def.1]{SchmuckRaisch2014_ControlLetters}, the system 
${\ES{}=\Tuple{\T,\WT\times\X,\BehS{}}}$ is an \textit{asynchronous state space dynamical system} if
\begin{equation}\label{equ:StateSpaceDynamicalSystem:2}
 \AllQ{\Tuple{w_1,x_1},\Tuple{w_2,x_2}\in\BehS{}, t_1,t_2\in\T}{\propImp*{x_1(t_1)=x_2(t_2)}{\CONCAT{\Tuple{w_1,x_1}}{t_1}{t_2}{\Tuple{w_2,x_2}}\in\BehS{}},}
\end{equation}
and we say that $\ES{}=\Tuple{\T,\WT\times\X,\BehS{}}$ is an asynchronous state space system for ${\E=\Tuple{\T,\WT,\Beh}}$ if ${\projState{W}{\BehS{}}=\Beh}$.\\
A \textit{state machine} is a tuple ${\Psys=(\X,\W,\tr,\Xo{})}$, where $\X
$ is the 
state space, $\W$ 
is the 
signal space, $\Xo{}\subseteq\X$ is the set of initial states and $\tr\subseteq\X\timesps\W\timesps\X$ is a next state relation.
Then 
\begin{equation*}
 \Beh_f(\Psys):=\SetCompX{\Tuple{\w,\x}}{
\begin{propConjA}
x(0)\in\Xo{}\\
 \AllQ{t\in\T}{\Tuple{\x(t),\w(t),\x(t+1)}\in\tr}
\end{propConjA}}
\end{equation*}
is the \textit{full behavior induced by $\Psys$}, and we say that $\Psys=\allowbreak\Tuple{\X,\W,\tr,\Xo{}}$ realizes ${\ES{}=\Tuple{\Nbn,\WT\times\X,\BehS{}}}$ if $\Beh_f(\Psys)=\BehS{}$. Furthermore, $\Psys$ is a \textit{finite} state machine if $\length{X}<\infty$ and $\length{W}<\infty$.\\
Now let ${\T=\Nbn}$. Then, given time instants $t_1,t_2\in\Nbn$, $t_1\leq t_2$, the string $w\in\W^{[t_1,t_2]}$ is of length $\lengthw{\w}=t_2-t_1+1$ and for $t_1< t_2$ we define
$\w\ll{t_2,t_1}:=\lambda$, where $\lambda$ denotes the \textit{empty string} with $\lengthw{\lambda}=0$. Furthermore, the concatenation of the restrictions ${w_1'=w_1\ll{0,t_1}}$ and  ${w_2'=w_2\ll{0,t_2}}$, with $t_1,t_2\in\Nbn,~t_1\leq t_2$, is defined as the standard concatenation of finite strings, i.e., 
$\w_1'\sconc\w_2':=\BR{\CONCAT{w_1}{t_1+1}{0}{w_2}}\hspace{-1mm}\ll{0,t_1+t_2+1}$. Furthermore, for a finite string $w=\nu_0\nu_1\hdots\nu_l$ we denote the restriction of $w$ by $w\lb{i,j}:=\nu_i\hdots\nu_j$ with ${0\leq i\leq j\leq l}$.

\section{$\signalmap$~-~Dynamical Systems}\label{sec:sys}
The common starting point of methods generating finite state abstractions of a possibly continuous-time and infinite state system is the definition of a finite external signal space $\Gamma$. This external signal space can be understood as the information content which needs to be preserved or approximated when interconnecting  the system to other components or when controlling it w.r.t. a given specification and is therefore application-specific.
While the evolution of the introduced external variable on $\Gamma$ is in discrete-time $\Nbn$, the internal dynamics will still evolve on the original, possibly continuous time axis $\T$.
To handle such models with distinct internal and external time axes, we use the notion of $\signalmap$-dynamical systems.\smalllb

\begin{definition}[\cite{SchmuckRaisch2014_HSCC}, Def.1] \label{def:DynSysInducesFiniteV}
Let ${\E=\Tuple{\T,\WT,\Beh}}$ be a dynamical system.
Then $\Ep{}=\EpRhsDisc$ is a \textbf{$\signalmap$-dynamical system} if
\begin{equation*}
 \signalmap:\Beh\fun\twoup{\Gamma^{\Nbn}\times\timescale}
\end{equation*}
where $\Gamma$ is an external signal space,
\[\timescale=\SetCompX{\timescaleUp{\tau}:\T\parfun\Nbn}{
 \timescaleUp{\tau} \text{ is surjective and}\text{monotonically increasing}}\] 
 is a set of time scale transformations and 
 \begin{equation}\label{equ:BehE}
  \BehE{}=\SetCompX{\gamma\in\Gamma^{\Nbn}}{
 \ExQ{w\in\Beh,\tau\in\timescale}{\Tuple{\gamma,\tau}\in\signalmap(\w)}}
 \end{equation}
 is the external behavior. Furthermore, ${\timescaleDown{\tau}:\Nbn\fun\twoup{\T}}$ denotes the inverse time scale transformation\footnote{If $\AllQ{k\in\Nbn}{\length{\timescaleDown{\tau}(k)}=1}$, by slightly abusing notation, we denote the unique element $t_k\in\timescaleDown{\tau}(k)$ by $\timescaleDown{\tau}(k)$ itself and write $t_k=\timescaleDown{\tau}(k)$.}, i.e., $\timescaleDown{\tau}(k)=\SetComp{t\in\T}{\timescaleUp{\tau}(t)=k}$.
\end{definition}

In a $\signalmap$-dynamical system the map $\signalmap$ describes how internal signals are discretized (in space and time) to generate the external behavior $\BehE{}$.
The concept covers both time-triggered and event-triggered discretization.
The following example illustrates how event-triggered discretization can be captured in a $\signalmap$-dynamical system.\smalllb

\begin{example}\label{exp:1}\normalfont
Consider a dynamical system ${\E=\Tuple{\T,\WT,\Beh}}$ with ${\T=\Rbn}$, ${\WT=\INTERSECT{\Rb}{[-10,10]}}$, and ${w\in\Beh}$ iff $w$ is continuous and $w(0)\in\Set{-10,10}$. 
 Using $\Gamma=\Set{m_2,m_1,p_1,p_2}$ and the sets
 \begin{align*}
  I_{m_2}&=[-10,-4),&I_{m_1}&=(-6,1),\\
  I_{p_1}&=(-1,6),&I_{p_2}&=(4,10],
 \end{align*}
 the external signals are constructed 
 using a set-valued discretization map $\mathfrak{d}:W\fun\twoup{\Gamma}$ \SUCHTHAT 
 \[\AllQ{G\in\Gamma,\nu\in\W}{\propAequ{G\in \mathfrak{d}(\nu)}{\nu\in I_{G}}}.\] %
Out of the many different options to construct $\signalmap$ from $\mathfrak{d}$, we discuss the two maps $\signalmap_a$ and $\signalmap_b$ as depicted in Fig.~\ref{fig:timescale1}~-~\ref{fig:timescale2}.\\
The signal map  $\signalmap_a$ is constructed s.t.
for all $ \gamma\in\Gamma^{\Nbn},\tau_a\in\timescale$ and $\w\in\Beh$, it holds that $\Tuple{\gamma,\tau_a}\in\signalmap_a(\w)$ \IFF
\begin{equation*}
 \gamma(0)\in \mathfrak{d}(w(0)),\quad\tau_a^{-1}(0)=\Set{0}
\end{equation*}
and for all $k\in\Nbn,k>0$ it holds that 
\begin{align}
 \tau_a^{-1}(k)&=\mathrm{glb}\SetCompX{t\geq\tau_a^{-1}(k-1)}{w(t)\notin\mathfrak{d}^{-1}(\gamma(k-1))}&\text{and}&&
 \gamma(k)&\in \mathfrak{d}(w(\tau_a^{-1}(k))),\label{equ:exp:1:1}
\end{align}
where $\mathrm{glb}$ denotes the greatest lower bound and $\AllQ{G\in\Gamma}{\mathfrak{d}^{-1}(G)=I_{G}}$. This construction generates a time scale transformation where different points in $\dom{\tau_a}$ are mapped to different points in $\Nbn$ as depicted in Fig.~\ref{fig:timescale1} (middle). We therefore call $\tau_a$ a \textit{point to point time scale transformation}.
$\signalmap_a$ triggers an external event when $w$ leaves its current interval, generating the external signal depicted in Fig.~\ref{fig:timescale2}.
It is easy to see that $\tau_a$ can be used to define a \textit{set to point time scale transformation} 
\begin{equation}\label{equ:exp:1:2}
 \tau_b^{-1}(k)=\left[\tau_a^{-1}(k),\tau_a^{-1}(k+1)\right),
\end{equation}
depicted in Fig.~\ref{fig:timescale1} (bottom).
Every point in $\T$ is in the domain of $\tau_b$.  Combining the construction of $\tau_b$ in \eqref{equ:exp:1:2} with the construction of $\gamma$ in \eqref{equ:exp:1:1} defines a signal map $\signalmap_b$.\\
The resulting $\signalmap$-dynamical systems $\Ep{i}=\ensuremath{(\T,\allowbreak\Nbn,\allowbreak\WT,\allowbreak\Gamma,\allowbreak\Beh,\allowbreak\BehE{},\allowbreak\signalmap_{i})},~i\in\Set{a,b}$ then only differ w.r.t. their timescale transformations included in $\signalmap_a$ and $\signalmap_b$. 
\end{example}
 \begin{figure}[htb]
 \begin{center}
 \begin{tikzpicture}[auto,scale=1]
 \begin{pgfonlayer}{background}
   \path      (0,0) node (o) {
      \includegraphics[width=0.525\linewidth]{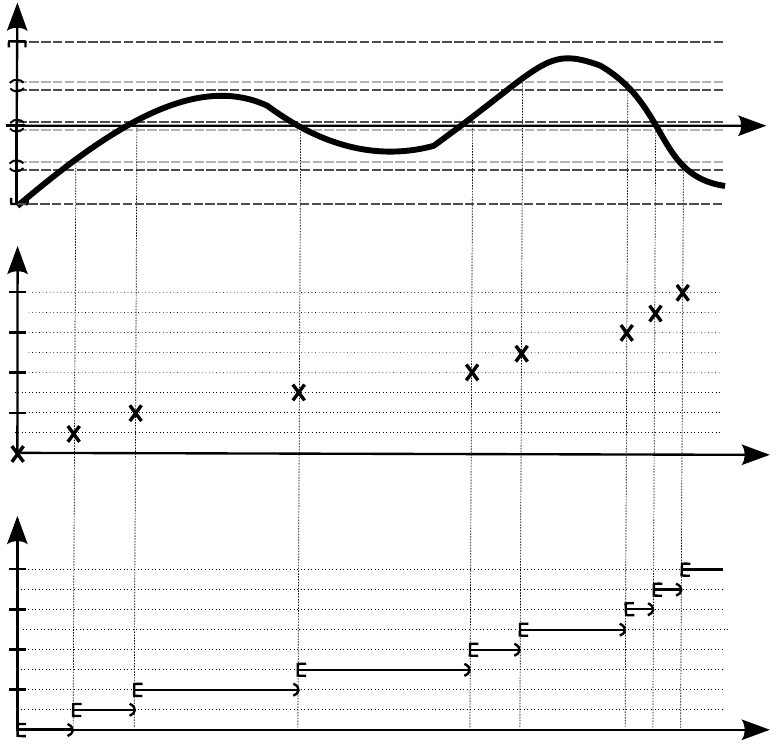}};
 \end{pgfonlayer}
 \begin{small}
  \begin{pgfonlayer}{foreground}
  \path (o.north west)+(-0.1,-0.2) node {$\W$};
  \path (o.north west)+(-0.1,-0.7) node {$10$};
  \path (o.north west)+(-0.1,-1.2) node {$5$};
  \path (o.north west)+(-0.1,-1.7) node (b) {$0$};
  \path (o.north west)+(-0.3,-2.2) node {$-5$};
  \path (o.north west)+(-0.3,-2.7) node  {$-10$};  
   \path (o.north east)+(0,-2.4) node (t2) {$w\inps\W^\T$};  
   \path (o.north east)+(0.1,-1.8) node  {$T$};
   \path (b)+(0.1,-1.6) node {$\Nbn$};
   \path (b)+(0.1,-2) node {$8$};    
   \path (b)+(0.1,-2.5) node {$6$};  
   \path (b)+(0.1,-3) node {$4$}; 
   \path (b)+(0.1,-3.5) node {$2$}; 
   \path (b)+(0.1,-4) node (c) {$0$};
   \path (t2)+(0,-1.2) node  {$\tau_a\inps\Nbn^\T$};  
   \path (t2)+(0,-3.4) node  (t3) {$T$};   
   \path (c)+(0,-0.8) node {$\Nbn$};
   \path (c)+(0,-1.3) node {$8$};    
   \path (c)+(0,-1.8) node {$6$};  
   \path (c)+(0,-2.3) node {$4$}; 
   \path (c)+(0,-2.8) node {$2$}; 
   \path (c)+(0,-3.3) node (c) {$0$};
   \path (t3)+(0,-1.2) node  {$\tau_b\inps\Nbn^\T$};  
   \path (t3)+(0,-3.3) node  {$T$};     
 \end{pgfonlayer}
  \end{small}
 \end{tikzpicture}
 \end{center}
 \caption{Illustration of point to point ($\tau_a$) and set to point ($\tau_b$) time scale transformations in Ex.\ref{exp:1}.}\label{fig:timescale1}
      \vspace{-0.5cm}
 \end{figure}
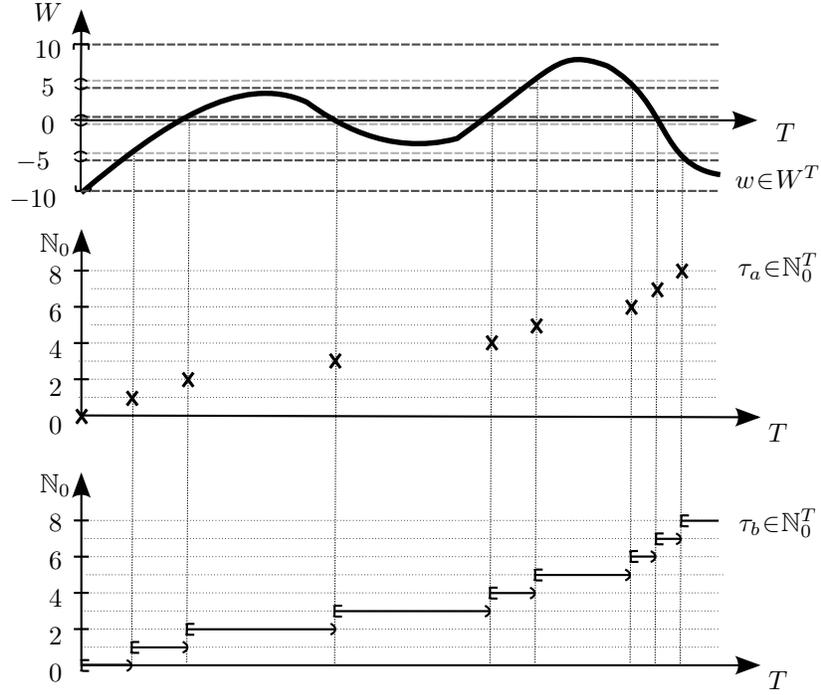
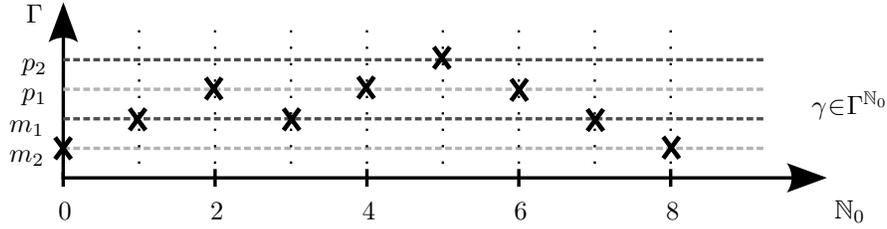
\begin{figure}[htb!]
\begin{center}
 \begin{tikzpicture}[auto,scale=1.01]
 \begin{pgfonlayer}{background}
   \path      (0,0) node (o) {
      \includegraphics[width=0.6\linewidth]{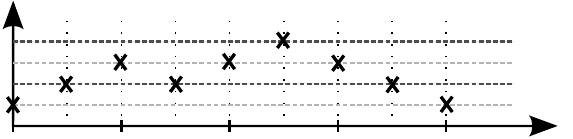}};
 \end{pgfonlayer}
  \begin{small}
 \begin{pgfonlayer}{foreground}
  \path (o.north west)+(0,-0.3) node (a) {$\Gamma$};
  \path (o.north west)+(0,-1) node (b) {$p_2$};
  \path (o.north west)+(0,-1.4) node (b) {$p_1$};
  \path (o.north west)+(-0.1,-1.8) node (b) {$m_1$};
  \path (o.north west)+(-0.1,-2.2) node (b) {$m_2$};  
    \path (o.north east)+(-0.1,-2.9) node (t1) {$\Nbn$};  
    \path (o.north east)+(-0.1,-1.5) node (t1) {$\gamma\inps\Gamma^{\Nbn}$};
  \path (o.north west)+(0.4,-2.9) node (b) {$0$}; 
  \path (o.north west)+(2.4,-2.9) node (b) {$2$}; 
  \path (o.north west)+(4.4,-2.9) node (b) {$4$};
  \path (o.north west)+(6.4,-2.9) node (b) {$6$};
  \path (o.north west)+(8.4,-2.9) node (b) {$8$};
 \end{pgfonlayer}
   \end{small}
 \end{tikzpicture}
 \end{center}
   \vspace{-0.7cm}
 \caption{Illustration of the external signal constructed using event triggered discretization in Ex.\ref{exp:1}.}\label{fig:timescale2}
 \end{figure}

States are internal variables for which the axiom of state holds, i.e., all relevant information on the past of the system is captured by those variables. As discussed in \cite{SchmuckRaisch2014_ControlLetters}, two variants of the state property exist for behavioral systems, a synchronous and an asynchronous one. The latter is characterized by \eqref{equ:StateSpaceDynamicalSystem:2}, the former by \eqref{equ:StateSpaceDynamicalSystem:2} if ${t=t_1=t_2}$.
It will be discussed later in Remark~\ref{rem:WhyAsyncSS} why we restrict attention to asynchronous state space systems. An asynchronous state space $\signalmap$-dynamical systems is a system where the asynchronous state property is preserved by the signal map $\signalmap$.

\begin{definition}[\cite{SchmuckRaisch2014_HSCC}, Def.2]\label{def:SSDynSysInducesFiniteV}
Let $\Ep{}=\EpRhsDisc$ be a $\signalmap$-dynamical system, $\X$ be a set and $\BehS{}\subseteq(\WT\times\X)^\T$. 
Then $\EpS{}=\EpSRhsDisc{}$ is an \textbf{asynchronous state space $\signalmap$-dynamical system} if
 \begin{equation}\label{equ:def:SSDynSysInducesFiniteV}
   \AllQSplit{\Tuple{w_1,x_1}\in\BehS{},\Tuple{w_2,x_2}\in\BehS{},t_1,t_2\in\T,\Tuple{\gamma_2,\tau_2}\in\signalmap(w_2),\Tuple{\gamma_1,\tau_1}\in\signalmap(w_1),k_1,k_2\in\Nbn}{\propImp{
   \begin{propConjA}
    x_1(t_1)=x_2(t_2)\\
    k_1=\timescaleUp{\tau_1}(t_1)\\
    k_2=\timescaleUp{\tau_2}(t_2)
   \end{propConjA}}{
\begin{propConjA}
 \CONCAT{\Tuple{w_1,x_1}}{t_1}{t_2}{\Tuple{w_2,x_2}}\in\BehS{}\\
\Tuple{\CONCAT{\gamma_1}{k_1}{k_2}{\gamma_2},\CONCAT{\tau_1}{t_1}{t_2}{\BR{\tau_2+c}}}\in\signalmap(\CONCAT{w_1}{t_1}{t_2}{w_2})
\end{propConjA}}}
\end{equation}
where $\AllQ{t\in\T}{c(t)=k_1-k_2}$.
Furthermore, $\EpS{}$ is an asynchronous state space $\signalmap$-dynamical system for $\Ep{}$, if ${\projState{\WT}{\BehS{}}=\Beh}$.
\end{definition} 
Since possibly not all states are reachable by a state trajectory in $\projState{\X}{\BehS{}}$, following \cite[Def.5]{SchmuckRaisch2014_HSCC}, we define reachable subsets of the state space.\smalllb

\begin{definition}\label{def:TimeIndexedStateSpace1}
Let $\EpS{}=\EpSRhsDisc{}$ be an asynchronous state space $\signalmap$-dynamical system. Then the \textbf{internal and external reachable state spaces} $\XT{}\subseteq\X$ and $\XK{}\subseteq\X$, respectively, are defined as 
\begin{align*}
 &\XT{}\deff\bigcup_{t\in\T}\Xt{t}{}\quad\text{and}\quad \XK{}\deff\bigcup_{k\in\Nbn}\Xk{k}{}\quad	\text{s.t.}\notag\\
 &\Xt{t}{}\deff\SetCompX{\xi}{\ExQ{\Tuple{w,x}\in\BehS{}}{x(t)=\xi}} \quad \text{and}\notag\\
 &\Xk{k}{}\deff\SetCompX{\xi}{
 \ExQ{\Tuple{w,x}\in\BehS{},\Tuple{\gamma,\tau}\in\signalmap(w),t\in\timescaleDown{\tau}(k)}{x(t)=\xi}
 }.
\end{align*}
Now let $\zeta$ be a finite string of symbols from $\Gamma$.
Then $\Xx{}{\zeta}:=\bigcup_{k\geq\lengthw{\zeta}}\Xxr{}{k}{\zeta}$ is the set of states compatible with a \enquote{recent past} $\zeta$ s.t. $\forall k\geq\lengthw{\zeta}$
\begin{equation}\label{equ:SetOfReachableStates}
\Xxr{}{k}{\zeta}=:\SetCompX{\xi}{
\ExQ{\Tuple{w,x}\in\BehS{},\Tuple{\gamma,\tau}\in\signalmap(w),t\in\timescaleDown{\tau}(k)}{\begin{propConjA}
                                          x(t)=\xi\\
					  \zeta=\gamma\ll{k-\lengthw{\zeta},k-1}
                                         \end{propConjA}}
}
\end{equation}
is the set of states reachable at time $t$ corresponding to external time $k$ and compatible with a \enquote{recent past} $\zeta$.
\end{definition}

Obviously $\Xxr{}{k}{\cdot}\subseteq\Xk{k}{}$.\\
Since the set $\Xx{}{\cdot}$ will be extensively used in the remainder of this paper, we illustrate its construction by an example.\smalllb

\begin{example}\label{exp:Xx}\normalfont
 Consider $\EpS{a}$ and $\EpS{b}$ constructed in Ex.\ref{exp:1} and assume $X=W$, i.e., signals $w\in\Beh$ can be asynchronously concatenated. Then, with $\BehS{}=\SetComp{\Tuple{x,w}}{\propConj{x=w}{w\in\Beh}}$, the systems $\EpS{i}=(\T,\allowbreak\Nbn,\allowbreak\WT\nobreak\times\nobreak\X,\allowbreak\Gamma,\allowbreak\BehS{},\allowbreak\BehE{},\allowbreak\signalmap_{i}),~i\in\Set{a,b},$ are asynchronous state space $\signalmap$-dynamical systems.\\
 First let $\zeta=\lambda$, i.e., $\lengthw{\zeta}=0$, then
 \begin{align*}
   \Xxr{a}{0}{\lambda}&=\Set{-10,10} &\Xxr{b}{0}{\lambda}&=\UNION{[-10,-4)}{(4,10]}
  \end{align*}
 are the sets of states with a \enquote{recent past} $\lambda$ reached at a time $t$ corresponding to external time $k=0$.
 Now consider strings $\zeta\in\Gamma=\Set{m_2,m_1,p_1,p_2}$, i.e., $\lengthw{\zeta}=1$, then
    \begin{align*}   
   \Xx{a}{m_2}&=\Set{-4}	&\Xx{b}{m_2}&=(-6,1)\\
    \Xx{a}{m_1}&=\Set{-6,1}	&\Xx{b}{m_1}&=\UNION{[-10,-4)}{(-1,6)}
  \end{align*}
    \begin{align*}
   \Xx{a}{p_1}&=\Set{-1,6}	&\Xx{b}{p_1}&=\UNION{(-6,1)}{(4,10]}\\
   \Xx{a}{p_2}&=\Set{4}		&\Xx{b}{p_2}&=(-1,4)
  \end{align*}
are the sets of states compatible with a \enquote{recent past} $G\in\Gamma$.
Observe that with a point to point time scale transformation only states reached at sampling instances are in $\Xx{a}{\cdot}$.
\end{example}

\section{$l$-complete Approximations}\label{sec:lcomp}
It was shown in \cite{SchmuckRaisch2014_ControlLetters} that asynchronous $l$-complete approx\-imations can be used to generate a finite state abstraction of a dynamical system, if it evolves on the discrete time axis $\Nbn$ and the external signal space is finite. If $\Gamma$ is finite, the external dynamical system $\EE{}=\Tuple{\Nbn,\Gamma,\BehE{}}$, with $\BehE{}$ as in \eqref{equ:BehE}, meets these requirements. Following \cite{MoorRaisch1999} and \cite{SchmuckRaisch2014_ControlLetters}, a system $\E'_E=\Tuple{\Nbn,\Gamma,\Beh'_E}$ is an  asynchronous $l$-complete approximation of $\EE{}$, if 
\begin{inparaenum}[(i)]
 \item  $\E'_E$ is asynchronously $l$-complete and
 \item  $\Beh'_E\supseteq\BehE{}$.
 \end{inparaenum}
 $\E'_E$ is a strongest asynchronous $l$-complete approximation of $\EE{}$ if 
 \begin{inparaenum}[(i)]
 \item it is an  asynchronous $l$-complete approximation of $\EE{}$ and
 \item for all other asynchronous $l$-complete approximations $\EE{}''=\Tuple{\Nbn,\Gamma,\BehE{}''}$ of $\EE{}$ it holds that $\Beh'_E\subseteq\BehE{}''$.
\end{inparaenum}
\smalllb

\begin{lemmaRef}[\cite{SchmuckRaisch2014_ControlLetters}, Lemma 7]\label{lem:constructElMax_general_TA}
 Let $\EE{}=\Tuple{\Nbn,\Gamma,\BehE{}}$ be a dynamical system and 
  \begin{equation*}\label{equ:lem:constructElMax_general_TA}
   \BehElaMax{}:=\SetCompX{\gamma\in\Gamma^{\Nbn}}{\begin{propConjA}
\gamma\ll{0,l}\in\BehE{}\ll{0,l}\\
 \AllQ{k\in\Nbn}{\gamma\ll{k,k+l}\inps \displaystyle\bigcup_{k'\in\Nbn}\hspace{-0.1cm}\BehE{}\ll{k',k'+l}}
\end{propConjA}}
.
 \end{equation*}
 Then ${\EElaMax{}:=\Tuple{\Nbn,\Gamma,\BehElaMax{}}}$ is the unique \textbf{strongest asyn\-chronous $l$-complete approximation} of $\EE{}$.
\end{lemmaRef}

Since $\EElaMax{}$ is asynchronously $l$-complete, we can use a trivial state space representation saving the last $l$ symbols of the external signal $\gamma\in\BehE{}$ in the current state.\smalllb

\begin{lemma}\label{lem:CorrPastIndSS}
Let $\EElaMax{}=\Tuple{\Nbn,\Gamma,\BehElaMax{}}$ be an asynchronous $l$-complete dynamical system,
\begin{equation*}
 \textstyle\Z:=\UNION{\BR{\bigcup_{r\in[0,l-1]}\BehElaMax{}\ll{0,r-1}}}{\BR{\bigcup_{k\in\Nbn}\BehElaMax{}\ll{k,k+l-1}}}
\end{equation*}
be a set 
and $\BehlaMaxS{}\subseteq(\Gamma\times\Z)^{\Nbn}$ \SUCHTHAT $\Tuple{\gamma,\z}\in\BehlaMaxS{}$ \IFF
\begin{equation*}\label{equ:lem:CorrPastIndSS}
 \z(k)=\begin{cases}
\gamma\ll{0,k-1}&0\leq k<l\\\gamma\ll{k-l,k-1}&k\geq l
\end{cases}\quad \text{and} \quad \gamma\in\BehElaMax{}.
\end{equation*}
Furthermore, let 
\begin{equation}\label{equ:signalmap_l}
 \philaMax\deff\BehElaMax{}\fun\twoup{\Gamma^{\Nbn}\times\timescale}~\text{s.t.}~\AllQ{\gamma\in\BehElaMax{}}{\philaMax(\gamma)=\Set{\gamma,\I}}.
\end{equation}
 Then
 \begin{inparaenum}[(i)]
  \item $\ElaMaxS{}:=\Tuple{\Nbn,\Gamma\times\Z,\BehlaMaxS{}}$ is an asynchronous state space system for $\EElaMax{}$ and 
  \item $\EplaMaxS{}:=\EplaMaxSRhs{}$ is an asynchronous $\signalmap$-dynamical state space system 
  for $\EElaMax{}$
 \end{inparaenum}. 
\end{lemma}

\begin{proof}
Part (i) is proven by \cite[Lemma 3, Lemma 6]{SchmuckRaisch2014_ControlLetters}. For the second part observe that 
 part (i) implies \eqref{equ:StateSpaceDynamicalSystem:2}. Now the trivial signal map $\philaMax$ immediately implies that also \eqref{equ:def:SSDynSysInducesFiniteV} holds, since\linebreak $\CONCAT{\I}{k_1}{k_2}{(\I+(k_1-k_2))}=\I$.
\end{proof}\vspace{0.1cm}

The state space system $\ElaMaxS{}$ constructed in \REFlem{lem:CorrPastIndSS} is a finite state abstraction of the external behavior of $\EpS{}$. Note that, by construction, $\BehElaMax{}\supseteq\Beh_{E}^{(l+1)^\uparrow}$. Hence the parameter $l$ can be used to adjust approximation accuracy.\\
Observe that $\EplaMaxS{}$ and $\ElaMaxS{}$ in \REFlem{lem:CorrPastIndSS} exhibit the same external behavior. $\EplaMaxS{}$ is the trivial transformation of $\ElaMaxS{}$ into the framework of $\signalmap$-dynamical systems. This construction is needed to formally relate the original system $\EpS{}$ to its finite state abstraction $\ElaMaxS{}$ using the framework of  
$\signalmap$-dynamical systems as discussed in \REFsec{sec:SimRel}-\ref{sec:SimRelApprox}.\\
As a main result from \cite{SchmuckRaisch2014_ControlLetters}, the abstraction $\ElaMaxS{}$ (and therefore also $\EplaMaxS{}$) can be realized by a finite state machine (FSM) if $\Gamma<\infty$.\smalllb

\begin{lemmaRef}[\cite{SchmuckRaisch2014_ControlLetters}, Lemma 6]\label{lem:realizationTA}
Given the premises of \REFlem{lem:CorrPastIndSS}, $\length{\Gamma}<\infty$, $\Z_0=\Set{\lambda}$ and 
 \begin{align*}\label{equ:delta_timeInv}
\tr=&\SetCompX{\Tuple{\zeta,\sigma,\zeta\sconc\sigma}}{\propConj{\lengthw{\zeta}<l}{\zeta\sconc\sigma\in\BehElaMax{}\ll{0,\lengthw{\zeta}}}}\\
&\cup\SetCompX{\Tuple{\zeta,\sigma,\zeta\lb{1,l-1}\sconc\sigma}}{\propConj{\lengthw{\zeta}=l}{\textstyle\zeta\sconc\sigma\in\bigcup_{k'\in\Nbn}\BehElaMax{}\ll{k',k'+l}}},
\end{align*}
the system $\ElaMaxS{}=\Tuple{\Nbn,\Gamma,\Z,\BehlaMaxS{}}$ is realized by the finite state machine $\Psys=(\Z,\Gamma,\tr,\Z_0)$.
\end{lemmaRef}
Note, that for time-variant systems, $\BehElaMax{}$ and therefore $\tr$ may not be computable.

\begin{example}\label{exp:Blmax}\normalfont
 Consider the $\signalmap$-dynamical systems $\EpS{a}$ and $\EpS{b}$ constructed in Ex.\ref{exp:1}  and let $l=1$. Then
 \begin{align*}
  \BehE{}\ll{0,1}&=\Set{m_2m_1,p_2p_1},\\
\bigcup_{k'\in\Nbn}\hspace{-1.2mm}\BehE{}\ll{k',k'+1}&=\Set{m_2m_1,m_1m_2,m_1p_1,p_1m_1,p_1p_2,p_2p_1}
 \end{align*}
 and we can construct $\Beh_{E}^{1^\uparrow}$ by playing the domino-game depicted in Fig.~\ref{fig:exp:Blmax:domino}, starting with the domino $m_2m_1$ or $p_2p_1$ and always appending a domino from the set $\bigcup_{k'\in\Nbn}\BehE{}\ll{k',k'+1}$ starting with the last symbol of the previous domino. Observe that for this simple example, $\Beh_{E}^{1^\uparrow}$ is actually identical to $\BehE{}$ implying that $\EE{}=\Tuple{\Nbn,\Gamma,\BehE{}}$ is asynchronous $1$-complete.
 Using \REFlem{lem:CorrPastIndSS} we can construct the state space for $\Sigma_{S}^{1^\uparrow}$ and obtain $Z=\Set{\lambda,m_2,m_1,p_1,p_2}$. 
$\Sigma_{S}^{1^\uparrow}$ is realized by the FSM depicted in Fig.~\ref{fig:exp:Blmax:FSM}.
\end{example}

\begin{figure}[htb]
\begin{minipage}{0.45\linewidth}
\begin{center}
 \begin{tikzpicture}[auto,scale=1.05]
 \begin{pgfonlayer}{background}
   \path      (0,0) node (o) {
      \includegraphics[width=0.5\linewidth]{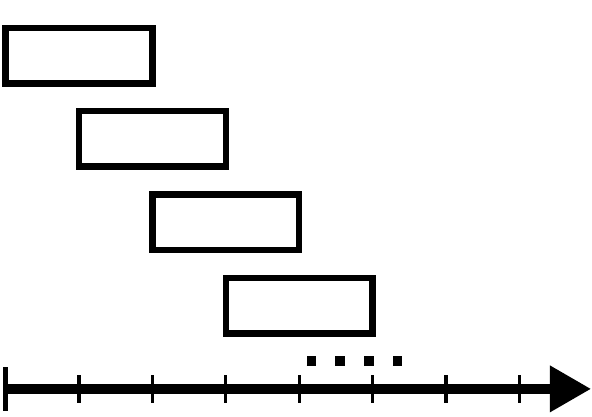}};
 \end{pgfonlayer}
 \begin{pgfonlayer}{foreground}
  \path (o.north west)+(0.55,-0.5) node (i) {${\scriptstyle~m_2~m_1}$};
  \path (i.north west)+(1.1,-0.74) node (ii){$\scriptstyle~m_1~p_1$};
  \path (ii.north west)+(1.1,-0.75) node (iii){$\scriptstyle~p_1~m_1$};
  \path (iii.north west)+(1.1,-0.75) node (iiii){$\scriptstyle~m_1~p_1$};
  \path (o.south west)+(0.15,-0.05) node (t0) {$0$};
  \path (t0)+(1,0) node (t2) {$2$};
  \path (t2)+(0.9,0) node (t4) {$4$};
  \path (t4)+(0.9,0) node (t6) {$6$};
  \path (t6)+(1,0) node (t6) {$\Nbn$};
 \end{pgfonlayer}
 \end{tikzpicture}
 \end{center}
    \vspace{-0.35cm}
 \caption{Domino game generating $\gamma=m_2m_1p_1m_1p_1\hdots$ (Ex.\ref{exp:Blmax}). }\label{fig:exp:Blmax:domino}

 \end{minipage}\hfill
 \begin{minipage}{0.50\linewidth}
 \begin{center}
  \begin{tikzpicture}[auto,scale=1.2]
\node (dummy) at (-0.7,0.5) {};
\node[state] (p0) at (0,0) {$\lambda$};
\node[state] (m_2) at (-1.5,-1) {$m_2$};
\node[state] (m_1) at (-0.5,-1) {$m_1$};
\node[state] (p_1) at (0.5,-1) {$p_1$};
\node[state] (p_2) at (1.5,-1) {$p_2$};
\SFSAutomatEdge{dummy}{}{p0}{}{}
\SFSAutomatEdge{p0}{m_2}{m_2.north}{bend right}{swap}
\SFSAutomatEdge{p0}{p_2}{p_2.north}{bend left}{}
\SFSAutomatEdge{m_2}{m_1}{m_1}{bend left}{}
\SFSAutomatEdge{m_1}{p_1}{p_1}{bend left}{}
\SFSAutomatEdge{p_1}{p_2}{p_2}{bend left}{}
\SFSAutomatEdge{p_2}{p_1}{p_1}{bend left}{}
\SFSAutomatEdge{p_1}{m_1}{m_1}{bend left}{}
\SFSAutomatEdge{m_1}{m_2}{m_2}{bend left}{}
%
\end{tikzpicture}
\end{center}
 \vspace{-0.2cm}
\caption{Finite state machine realizing $\Sigma_{S}^{1^\uparrow}$ (Ex.\ref{exp:Blmax}).}\label{fig:exp:Blmax:FSM}

 \end{minipage}
 \end{figure}

\section{Simulation Relations for $\signalmap$-Dynamical Systems}\label{sec:SimRel}
To be able to investigate the existence of simulation relations between the constructed finite state abstraction $\EplaMaxS{}$ and the original system $\EpS{}$, we review the simulation relations for $\signalmap$-dynamical systems introduced in \cite{SchmuckRaisch2014_HSCC}.
\smalllb

\begin{definition}[\cite{SchmuckRaisch2014_HSCC}, Def.4,6]\label{def:SimRel}
 Let $\EpS{i}=\EpSRhsDisc{i},~i\in{1,2}$ be state space $\signalmap$-dynamical systems.
Then a relation $\R\subseteq\X_1\times\X_2$ is an \textbf{asynchronous simulation relation} from $\EpS{1}$ to $\EpS{2}$ (written $\R\in\SR{\async}{\EpS{1}}{\EpS{2}}$) if
\begin{subequations}\label{equ:def:SimRel}
 \begin{equation}\label{equ:def:SimRel:a}
  \AllQ{\xi_1\in\XK{1}}{\ExQ*{\xi_2\in\XK{2}}{\Tuple{\xi_1,\xi_2}\in\R}}
 \end{equation}
and
 \begin{equation}\label{equ:def:SimRel:b}
 \AllQSplit{\Tuple{w_1,x_1}\in\BehS{1},\Tuple{w',x'}\in\BehS{2}, \Tuple{\gamma_1,\tau_1}\in\signalmap_1(w_1),\Tuple{\gamma',\tau'}\in\signalmap_2(w'), t_1\in\T_1,t_2\in\T_2, k_1,k_2\in\Nbn}{
\propImp{
\begin{propConjA}
\Tuple{x_1(t_1),x'(t_2)}\in\R\\
k_1=\timescaleUp{\tau_1}(t_1)\\
k_2=\timescaleUp{\tau'}(t_2)
\end{propConjA}
}{
\ExQSplit{\Tuple{w_2,x_2}\in\BehS{2}, \Tuple{\gamma_2,\tau_2}\in\signalmap_2(w_2)}{
\begin{propConjA}
\gamma_2=\CONCAT{\gamma'}{k_2}{k_1}{\gamma_1}\\
\AllQ{t\in\T_2, t< t_2}{
\begin{propConjA}
w_2(t)=w'(t)\\
x_2(t)=x'(t)\\
\tau_2(t)=\tau'(t)\\
\end{propConjA}
}\\
x_2(t_2)=x'(t_2)\\
\AllQSplit{k\geq k_2,t_1'\inps\timescaleDown{\tau_1}(k-k_2+k_1),t_1'>t_1}{\ExQ{t_2'\inps\timescaleDown{\tau_2}(k),t_2'\gps t_2 }{
\Tuple{x_1(t_1'),x_2(t_2')}\inps\R}}
\end{propConjA}
}
}}   
\end{equation}

\end{subequations}%
and an \textbf{$l$-initial simulation relation} from $\EpS{1}$ to $\EpS{2}$ (written $\R\in\SR{l}{\EpS{1}}{\EpS{2}}$) if
\begin{subequations}%
 \begin{equation}\label{equ:def:InitSimRel}
  \AllQ{\xi_1\in\Xk{l}{1}}{\ExQ*{\xi_2\in\Xk{l}{2}}{\Tuple{\xi_1,\xi_2}\in\R}}
 \end{equation}
  \end{subequations}
and \eqref{equ:def:SimRel:b} holds.
Furthermore, it is an \textbf{externally synchronous simulation relation} from $\EpS{1}$ to $\EpS{2}$ (written $\R\in\SR{\wsync}{\EpS{1}}{\EpS{2}}$) if
\begin{subequations}\label{equ:def:SimRelSync}
\begin{equation}\label{equ:def:SimRelSync:a}
  \AllQ{k\in\Nbn,\xi_1\in\Xk{k}{1}}{\ExQ*{\xi_2\in\Xk{k}{2}}{\Tuple{\xi_1,\xi_2}\in\R}}
 \end{equation}
 \end{subequations}
and \eqref{equ:def:SimRel:b} holds for $k=k_1=k_2$.
Finally, if $\T=\T_1=\T_2$, then $\R$ is a \textbf{ synchronous simulation relation} from $\EpS{1}$ to $\EpS{2}$ (written $\R\in\SR{\sync}{\EpS{1}}{\EpS{2}}$) if 
\begin{subequations}\label{equ:def:SimRelSync:2}
\begin{equation}\label{equ:def:SimRelSync:2:a}
  \AllQ{t\in\T,\xi_1\in\Xt{t}{1}}{\ExQ*{\xi_2\in\Xt{t}{2}}{\Tuple{\xi_1,\xi_2}\in\R}}
 \end{equation}
 \end{subequations}
and \eqref{equ:def:SimRel:b} holds for $k=k_1=k_2$ and $t=t_1=t_2$.
\end{definition}

Analogously to \cite[Def.5]{SchmuckRaisch2014_HSCC} we can define four types of simulations and bisimulations.\smalllb

\begin{definition}[\cite{SchmuckRaisch2014_HSCC}, Def.5]\label{def:SimilarBisimilar}
Let $\EpS{1}$ and $\EpS{2}$ be state space $\signalmap$-dynamical systems.
Then $\EpS{1}$ is 
\begin{inparaenum}[(i)]\item asynchronously, \item externally synchronously, \item synchronously, and \item $l$-initially \end{inparaenum}
 \textbf{simulated}  by $\EpS{2}$, denoted by \begin{inparaenum}[(i)]\item ${\EpS{1}\kg{\async}\EpS{2}}$, \item ${\EpS{1}\kg{\wsync}\EpS{2}}$, \item ${\EpS{1}\kg{\sync}\EpS{2}}$, and \item ${\EpS{1}\kg{l}\EpS{2}}$, \end{inparaenum} respectively, if there exists an \begin{inparaenum}[(i)]\item asynchronous, \item externally synchronous, \item synchronous, and \item $l$-initial \end{inparaenum} simulation relation from $\ES{1}$ to $\ES{2}$, respectively.\\
Furthermore,
$\EpS{1}$ and $\EpS{2}$ are \begin{inparaenum}[(i)]\item asynchronously, \item externally synchronously, \item synchronously, and \item $l$-initially \end{inparaenum} \textbf{bisimilar}, denoted by \begin{inparaenum}[(i)]\item ${\EpS{1}\hg{\async}\EpS{2}}$, \item ${\EpS{1}\hg{\wsync}\EpS{2}}$, \item ${\EpS{1}\hg{\sync}\EpS{2}}$,  and \item ${\EpS{1}\hg{l}\EpS{2}}$, \end{inparaenum} respectively, if there exists a relation $\R\subseteq\X_1\times\X_2$ \SUCHTHAT $\R$ and $\R^{-1}=\SetComp{\Tuple{x_2,x_1}}{\Tuple{x_1,x_2}\in\R}$ are \begin{inparaenum}[(i)]\item asynchronous, \item externally synchronous, \item synchronous , and \item $l$-initial \end{inparaenum} simulation relations from $\ES{1}$ to $\ES{2}$ and from $\ES{2}$ to $\ES{1}$, respectively.
\end{definition}

In contrast to the asynchronous simulation relation, it was shown in \cite{SchmuckRaisch2014_HSCC} that a $0$-initial simulation relation
is also an (externally) synchronous simulation relation and implies behavioral inclusion.\smalllb

\begin{lemmaRef}[\cite{SchmuckRaisch2014_HSCC}, Lemma 1,~Thm.1]\label{lem:LinitSimImpliesOtherSims}
 Let $\EpS{1}$ and $\EpS{2}$ be state space $\signalmap$-dynamical systems. Then
 \begin{compactenum}[(i)]
   \item $\propImp{\R\in\SR{0}{\EpS{1}}{\EpS{2}}}{
   \R\in\SR{\wsync}{\EpS{1}}{\EpS{2}}}$
  \item $\propImp{\R\in\SR{0}{\EpS{1}}{\EpS{2}}}{
   \R\in\SR{\async}{\EpS{1}}{\EpS{2}}}$
  \item $\propImp{
  \begin{propConjA}
  \R\in\SR{0}{\EpS{1}}{\EpS{2}}\\
  \T_1=\T_2=\Nbn\\
  \AllQ{w_1,\Tuple{\gamma_1,\tau_1}\hspace{-0.1cm}\in\hspace{-0.1cm}\signalmap_1(w_1)}{\tau_1\hspace{-0.1cm}=\hspace{-0.1cm}\I}\\
  \AllQ{w_2,\Tuple{\gamma_2,\tau_2}\hspace{-0.1cm}\in\hspace{-0.1cm}\signalmap_2(w_2)}{\tau_2\hspace{-0.1cm}=\hspace{-0.1cm}\I}
  \end{propConjA}
  }{\R\hspace{-0.08cm}\in\hspace{-0.08cm}\SR{\sync}{\EpS{1}}{\EpS{2}}.}$
 \item ${\propImp{\BR{\ES{1}\kg{l=0}\ES{2}}}{\BR{\BehE{1}\subseteq\BehE{2}}}}$ and
\item ${\propImp{\BR{\ES{1}\hg{l=0}\ES{2}}}{\BR{\BehE{1}=\BehE{2}}}}$.
\end{compactenum}
\vspace{-0.3cm}
\end{lemmaRef}

\section{Relating the Original System and its Approximation}\label{sec:SimRelApprox}

Now we investigate the existence of simulation and bisimulation relations between an asynchronous state space $\signalmap$-dynamical system and its strongest asynchronous $l$-complete approximation. Using the results of \REFlem{lem:LinitSimImpliesOtherSims} we first construct a $0$-initial and an $l$-initial simulation relation from the original system to its abstraction. \smalllb                                       

\begin{lemma}\label{lem:ElmaxSpiSimulatesEs}
Let $\EpS{}=\EpSRhsDisc{}$ be an asynchronous state space $\signalmap$-dynamical system and \linebreak ${\EplaMaxS{}=\EplaMaxSRhs{}}$ an asynchronous $\signalmap$-dynamical state space system for the strongest asynchronous $l$-complete approximation 
$\ElMaxE{}$
of the discrete external dynamics $\EE{}=\Tuple{\Nbn,\Gamma,\BehE{}}$, constructed in \REFlem{lem:CorrPastIndSS}. Furthermore, let
 \begin{align}
 \R_0&=\SetCompX{\Tuple{\xi,\zeta}\inps\CARTPROD*{\X}{\Z}}{
\xi\inps
\DiCases
{\Xxr{}{\lengthw{\zeta}}{\zeta}}{\lengthw{\zeta}<l}
{\Xx{}{\zeta}}{\lengthw{\zeta}=l}
}\label{equ:lem:ElmaxSpiSimulatesEs:R0}\\
 \R_l&=\SetCompX{\Tuple{\xi,\zeta}\in\CARTPROD*{\X}{\Z}}{
\propConj{\lengthw{\zeta}=l}{\xi\in\Xx{}{\zeta}}}.\label{equ:lem:ElmaxSpiSimulatesEs:Rl}
\end{align}
be two relations.
Then
\begin{compactenum}[(i)]
 \item $\R_0\hspace{-0.08cm}\in\hspace{-0.08cm}\SR{0}{\EpS{}}{\EplaMaxS{}}$ and
\item $\R_l\hspace{-0.08cm}\in\hspace{-0.08cm}\SR{l}{\EpS{}}{\EplaMaxS{}}$.\smalllb  
\end{compactenum}
\end{lemma}

\begin{proof}
 See Appendix.
\end{proof}\vspace{0.1cm}

\begin{example}\label{exp:Rl}\normalfont
 Consider the $\signalmap$-dynamical systems $\EpS{i},~i\in\Set{a,b}$ introduced in Ex.\ref{exp:1}, their reachable state sets $\Xx{i}{\cdot}$ determined in Ex.\ref{exp:Xx} and the strongest asynchronous $1$-complete approximation  $\Sigma_{S}^{\signalmap,1^\uparrow}$ constructed in Ex.\ref{exp:Blmax}. Using \REFlem{lem:ElmaxSpiSimulatesEs} we know that 
   \begin{align*}
  \R_{1,a}=&\left\{
  \Tuple{-4,m_2},\Tuple{-6,m_1},\Tuple{1,m_1},
  \Tuple{-1,p_1},\Tuple{6,p_1},\Tuple{4,p_2}\right\},\\
  \R_{1,b}=&
\SetCompX{\Tuple{\xi,m_2}}{\xi\inps(-6,1)}
  \cup \SetCompX{\Tuple{\xi,m_1}}{\xi\inps[-10,-4)\cup(-1,6)}\cup\\
  &\SetCompX{\Tuple{\xi,p_1}}{\xi\inps(-6,1)\cup(4,10]}\cup\SetCompX{\Tuple{\xi,p_2}}{\xi\inps(-1,4)}
 \end{align*}
 are $1$-initial simulation relations from $\EpS{i}$ to $\Sigma_{S}^{\signalmap,1^\uparrow}$.
\end{example}

Since $l$-complete systems have the same external behavior as their strongest asynchronous $l$-complete approximations, i.e., ${\BehE{}=\BehlMaxE{}}$,
we could guess that the inverse relations of \eqref{equ:lem:ElmaxSpiSimulatesEs:R0} and 
\eqref{equ:lem:ElmaxSpiSimulatesEs:Rl} are $0$- and $l$-initial simulation relations from the abstraction to the original system, if $\EE{}$ is $l$-complete. However, for $\R_0$, observe
that the \enquote{recent past} of states $\xi$ reached at time $k<l$, i.e. $\xi\in\Xk{k<l}{}$, has length $k<l$ and is therefore, in general, not sufficient to uniquely determine the future behavior of an $l$-complete system.  
Furthermore, even for $\R_l$, $l$-completeness of $\EE{}$ is not sufficient, as the following example illustrates. \smalllb

\begin{example}\label{exp:BehX}\normalfont
For simplicity consider a discrete system $\E=\Tuple{\Nbn,\Gamma,\Beh_1},~\Gamma=\Set{a,b,c}$ realized by the FSM $\Psys_1$ depicted in Fig.~\ref{fig:exp:BehX:Abs} (left). Observe that $\E$ is asynchronously $1$-complete and its strongest asynchronous $1$-complete approximation $\Sigma^{1^\uparrow}$, realized by the FSM $\Psys_2$ depicted in Fig.~\ref{fig:exp:BehX:Abs} (right), has the same behavior. 
We can easily generate systems
 \begin{align*}
 \EpS{1}&\equps(
 \Nbn,\allowbreak\Nbn,\Gamma\timesps\Set{\xi_1,\xi_2,\xi_3},\Gamma,\Beh_f(\Psys_1),\Beh_1,\signalmap_l
 )\\
 \Sigma_{S}^{\signalmap,1^\uparrow}&\equps(\Nbn,\Nbn,\Gamma\timesps\Gamma,\Gamma,\Beh_f(\Psys_2),\Beh_1^{l^\uparrow},\signalmap_l)
 \end{align*}
from $\E$ and $\Sigma^{1^\uparrow}$, respectively, by using the trivial signal map $\signalmap_l$ \eqref{equ:signalmap_l}.
Now using \eqref{equ:lem:ElmaxSpiSimulatesEs:Rl} gives the relation
$\R_{l=1}^{-1}=\Set{\Tuple{a,\xi_1},\Tuple{a,\xi_3},\Tuple{b,\xi_2},\Tuple{c,\xi_2}}.$
It can be easily verified that $\R_{l=1}^{-1}$ is not a $1$-initial simulation relation from  $\Sigma_{S}^{\signalmap,1^\uparrow}$ to $\EpS{1}$ since $b$ and $c$ can occur in state $a$ of $\Psys_2$ while only $c$ can occur in state $\xi_3$ of $\Psys_1$.
\end{example}
\begin{figure}[htb]
\begin{center}
 \vspace{-0.2cm}
  \begin{tikzpicture}[auto,scale=1.4]
\node (dummy) at (0,0.7) {};
\node[state] (m_2) at (-1,0) {$\xi_1$};
\node[state] (m_1) at (0,0) {$\xi_2$};
\node[state] (p_1) at (1,0) {$\xi_3$};
\SFSAutomatEdge{dummy}{}{m_1}{}{}
\SFSAutomatEdge{m_1}{a}{m_2}{bend left}{}
\SFSAutomatEdge{m_2}{b}{m_1}{bend left}{}
\SFSAutomatEdge{m_1}{a}{p_1}{bend left}{}
\SFSAutomatEdge{p_1}{c}{m_1}{bend left}{}
%
\end{tikzpicture}
\hspace{2cm}
  \begin{tikzpicture}[auto,scale=1.4]
\node (dummy) at (-2.7,0) {};
\node[state] (q0) at (-2,0) {$\lambda$};
\node[state] (m_2) at (0,0) {$a$};
\node[state] (m_1) at (-1,0) {$c$};
\node[state] (p_1) at (1,0) {$b$};
\SFSAutomatEdge{dummy}{}{q0}{}{}
\SFSAutomatEdge{q0}{a}{m_2.north}{bend left}{pos=0.2}
\SFSAutomatEdge{m_2}{c}{m_1}{bend left}{}
\SFSAutomatEdge{m_1}{a}{m_2}{bend left}{}
\SFSAutomatEdge{m_2}{b}{p_1}{bend left}{}
\SFSAutomatEdge{p_1}{a}{m_2}{bend left}{}
%
\end{tikzpicture}
\end{center}
 \vspace{-0.2cm}
\caption{FSM $\Psys_1$ (left) and $\Psys_2$ (right) in Ex.\ref{exp:BehX}.}\label{fig:exp:BehX:Abs}
 \vspace{-0.2cm}
 \end{figure}

 We therefore have to additionally ensure that all states with identical \enquote{recent past} allow for the same future external behavior. Inspired by \cite[Thm.~4.18]{TabuadaBook} we formulate this property as an $l$-initial simulation relation $\R_{\mathcal{X}}$ from the original system to itself. The following lemma shows that this condition together with asynchronous $l$-completeness of the original system (i.e., ${\BehE{}=\BehlMaxE{}}$) is necessary and sufficient for $\R_l^{-1}$ to be a simulation relation from the  abstraction to the original system. \smalllb

\begin{lemma} \label{lem:EsLSimulatesElmaxSpi}
Given the premises of \REFlem{lem:ElmaxSpiSimulatesEs}
and 
 \begin{align}
 \R_{\mathcal{X}}&=\SetCompX{\Tuple{\xi_a,\xi_b}\inps\CARTPROD*{\X}{\X}}{
\ExQ{\zeta\in\Gamma^{l}}{
   \xi_a,\xi_b\hspace{-0.08cm}\in\hspace{-0.08cm}\Xx{}{\zeta}
}}\label{equ:Rx}
\end{align}
it holds that
\[\propAequ{
\begin{propConjA}
\BehE{}=\BehlMaxE{}\\
 \R_{\mathcal{X}}\in\SR{l}{\EpS{}}{\EpS{}}
\end{propConjA}
}{\R_l^{-1}\in\SR{l}{\EplaMaxS{}}{\EpS{}}}.\]
\end{lemma}

\begin{proof}
 See Appendix.
\end{proof}\vspace{0.1cm}

\begin{example}\label{exp:Rlt}\normalfont
 The inverse relations of $\R_{1,a}$ and $\R_{1,b}$ from Ex.\ref{exp:Rl} are a $1$-initial simulation relation from $\Sigma_{S}^{\signalmap,1^\uparrow}$ to $\EpS{i},~i\in\Set{a,b}$, respectively, since $\EpS{i}$ is $1$-complete and $\R_{\mathcal{X}}\in\SR{1}{\EpS{i}}{\EpS{i}}$ holds.
\end{example}

\begin{remark}
 Observe, that \eqref{equ:def:InitSimRel} and the last line of \eqref{equ:def:SimRel:b} holds for $\R_{\mathcal{X}}$ by construction, as shown in the proof of \REFlem{lem:EsLSimulatesElmaxSpi}. 
 Therefore, requiring $\R_{\mathcal{X}}\in\SR{l}{\EpS{}}{\EpS{}}$ only ensures, that all state trajectories starting with the same \enquote{recent past} are able to generate the same future external behavior by not changing their \enquote{full past}.
\end{remark}

\begin{remark}
\REFlem{lem:EsLSimulatesElmaxSpi} shows that the existence of an $l$-initial simulation relation from $\EplaMaxS{}$ to $\EpS{}$ implies external behavioral equivalence, i.e., ${\BehE{}=\BehlMaxE{}}$ (extending \REFlem{lem:LinitSimImpliesOtherSims}).
\end{remark}

 \begin{remark}\label{rem:WhyAsyncSS}
  It was shown in \cite{SchmuckRaisch2014_ControlLetters} that an asynchronous (in contrast to a synchronous) $l$-complete approximation can be represented by a finite state machine. This is necessary to apply well-known controller synthesis methods (e.g., SCT \cite{RamWon1984}). Since the focus of this paper is to construct a finite state abstraction for controller synthesis, we have restricted our attention to asynchronous $l$-complete approximations.\\
  It is easy to show, that the external behavior $\BehE{}$ of (externally) synchronous state space $\signalmap$-dynamical systems (see \cite[Def.2]{SchmuckRaisch2014_ControlLetters}) is asynchronously $l$-complete, i.e., $\BehE{}=\BehlMaxE{}$, if and only if the system is an asynchronous state space $\signalmap$-dynamical system. Therefore, we are only able to establish the results in \REFlem{lem:EsLSimulatesElmaxSpi} for the latter system class. 
\end{remark}

As our main result, we now show that a strongest asynchronous $l$-complete approximation simulates the original system in various ways and that the conditions in \REFlem{lem:EsLSimulatesElmaxSpi} imply bisimilarity of the original system and its approximation.\smalllb

\begin{theorem}\label{thm:EqRel_PISS}
Given the premises of \REFlem{lem:ElmaxSpiSimulatesEs}~-~\ref{lem:EsLSimulatesElmaxSpi} it holds that
\begin{enumerate}[\bfseries (i)]
 \item ${\EpS{}\kg{l=0}\EplaMaxS{}}$,
  \item ${\EpS{}\kg{\wsync}\EplaMaxS{}}$,
 \item ${\EpS{}\kg{\async}\EplaMaxS{}}$,\linebreak
 \item $\propImp{
  \begin{propConjA}
  \T=\Nbn\\
  \AllQ{w,\Tuple{\gamma,\tau}\in\signalmap(w)}{\tau=\I}
  \end{propConjA}
  }{\BR{\EpS{}\kg{\sync}\EplaMaxS{}}}$,
\item ${\EpS{}\kg{l}\EplaMaxS{}}$, and
\item $\propAequ{
\begin{propConjA}
{\BehE{}=\BehlMaxE{}}\\
 \R_{\mathcal{X}}\inps\SR{l}{\EpS{}}{\EpS{}}
\end{propConjA}
}{\propImp{
\begin{propConjA}
 \R_l\inps\SR{l}{\EpS{}}{\EplaMaxS{}}\\
 \R_l^{-1}\inps\SR{l}{\EplaMaxS{}}{\EpS{}}
\end{propConjA}
}{
\BR{\EpS{}\hg{l}\EplaMaxS{}}
}}.$\\
\end{enumerate}
\end{theorem}

\begin{proof}\label{proof:thm:EqRel_PISS}
 \begin{inparaenum}[\bfseries (i)]
  \item From \REFlem{lem:ElmaxSpiSimulatesEs} (i) and \REFdef{def:SimilarBisimilar}.
  \item From (i), \REFlem{lem:LinitSimImpliesOtherSims} (i) and \REFdef{def:SimilarBisimilar}.
  \item From (i), \REFlem{lem:LinitSimImpliesOtherSims} (ii) and \REFdef{def:SimilarBisimilar}.
  \item From (i), \REFlem{lem:LinitSimImpliesOtherSims} (iii) and \REFdef{def:SimilarBisimilar}. 
  \item From \REFlem{lem:ElmaxSpiSimulatesEs} (ii) and \REFdef{def:SimilarBisimilar}.
  \item From \REFlem{lem:ElmaxSpiSimulatesEs} (ii), \REFlem{lem:EsLSimulatesElmaxSpi} and \REFdef{def:SimilarBisimilar}
 \end{inparaenum}.
\end{proof}\vspace{0.1cm}

\begin{remark}\label{rem:ConnectionTabuada}
 There is a strong connection between \REFthm{thm:EqRel_PISS} (iv) 
 and the work of Tabuada \cite[Thm. 4.18]{TabuadaBook}. 
  It is possible to show that the construction of 
 the quotient system in \cite[Def.4.17]{TabuadaBook} coincides with a realization of the strongest asynchronous one-complete 
 approximation if a time-shifted version of the state space construction is used in \REFdef{lem:CorrPastIndSS} (analogously to \cite[Sec. 4]{Raisch2010}).
  This time-shift also \enquote{shifts} the definition of the sets $\Xx{}{\cdot}$ and therefore, all relations used in \REFlem{lem:ElmaxSpiSimulatesEs} and \ref{lem:EsLSimulatesElmaxSpi}. 
Using these time-shifted definitions, it can be shown, that for transition systems the results in \cite[Thm. 4.18]{TabuadaBook} and \REFthm{thm:EqRel_PISS} (iv) are equivalent.
\end{remark}

\section{Conclusion}\label{sec:conclusion}
We have shown in this paper, that the concepts of asynchronous state space $\signalmap$-dynamical systems and strongest asynchronous $l$-complete approximations can be combined to generate a finite state abstraction realizable by a finite state machine and therefore suitable for controller synthesis using supervisory control theory. Using simulation relations developed for $\signalmap$-dynamical systems, 
we have proven that a strongest asynchronous $l$-complete approximation simulates the original system in various ways. 
 In particular, necessary and sufficient conditions for the existence of a bisimulation relation where derived. It was discussed in Remark~\ref{rem:ConnectionTabuada}, that these conditions can be interpreted as a generalization of the results by Tabuada \cite[Thm. 4.18]{TabuadaBook} to abstractions with $l>1$ and systems which are not realizable by transition systems. We are currently preparing a paper where this connection is formally proven.

\begin{appendix}
 
\textit{Proof of \REFlem{lem:ElmaxSpiSimulatesEs} (i):}\\
\textbf{1)} \textbf{Show \eqref{equ:def:InitSimRel} holds} for ${l=0}$: 
\begin{compactitem}
\item Observe that $\Zk{0}{}=\Zt{0}{}=\Set{\lambda}$ (from \REFdef{def:TimeIndexedStateSpace1} (with state space $Z$) 
and the construction of $Z$ in \REFlem{lem:CorrPastIndSS}) and  $\Xxr{}{0}{\lambda}=\Xk{0}{}$ (from \eqref{equ:SetOfReachableStates} 
and \REFdef{def:TimeIndexedStateSpace1}).
\item With the construction of $\R_0$ in
\eqref{equ:lem:ElmaxSpiSimulatesEs:R0}, this implies $\AllQ{\xi\in\Xk{0}{}}{\Tuple{\xi,\lambda}\in\R_0}$, 
i.e., \eqref{equ:def:InitSimRel} holds for ${l=0}$.\smalllb
\end{compactitem}
\textbf{2)} \textbf{Show \eqref{equ:def:SimRel:b} holds}: 
\begin{compactitem}
 \item Using the trivial signal map $\philaMax$ we fix $\Tuple{w_1,x_1}\in\BehS{},\Tuple{\gamma',z'}\in\BehlaMaxS{}, \Tuple{\gamma_1,\tau_1}\in\signalmap(w_1), t_1\in\T, k_1,k_2\in\Nbn$ s.t. the right side of the implication in \eqref{equ:def:InitSimRel} holds, i.e.,  $k_1=\timescaleUp{\tau_1}(t_1)$ and $\Tuple{x_1(t_1),z'(k_2)}\in\R_0$.
 \item Now we construct particular $z_2$ and $\gamma_2$ to show that the right side of the implication in \eqref{equ:def:SimRel:b} holds:\\
\begin{inparaitem}[$\blacktriangleright$\hspace{-0.15cm}]
 \item With the construction of $\R_0$ in \eqref{equ:lem:ElmaxSpiSimulatesEs:R0} we have
 \[x_1(t_1)\in\DiCases
{\Xxr{}{\lengthw{z'(k_2)}}{z'(k_2)}}{\lengthw{z'(k_2)}<l}
{\Xx{}{z'(k_2)}}{\lengthw{z'(k_2)}=l}\]
and therefore we can fix (from \eqref{equ:SetOfReachableStates})
 $\Tuple{\wt,\xt}\in\BehS{},\Tuple{\gt,\taut}\in\signalmap(\wt),\kt\in\Nbn,\tet\in\timescaleDown{\tau}(\kt)$
 s.t. $\x_1(t_1)=\xt(\tet)$ and 
 $z'(k_2)=\gt\ll{\max(0,\kt-\lengthw{z'(k_2)}),\kt-1}$. 
 Furthermore, if $\lengthw{z'(k_2)}<l$ we have $\kt=\lengthw{z'(k_2)}$.\\
\item Since $\EpS{}$ is an asynchronous state space dynamical system, $\x_1(t_1)=\xt(\tet)$ implies that we can pick $w_1'=\CONCAT{\wt}{\tet}{t_1}{w_1}$,$x_1'=\CONCAT{\xt}{\tet}{t_1}{x_1}$,$\tau_1'=\CONCAT{\taut}{\tet}{t_1}{(\tau_1+\tilde{k}-k_1)}$ and $\gamma_1'=\CONCAT{\gt}{\kt}{k_1}{\gamma_1}$ and have $(w_1',x_1')\in\BehS{}$ and $\Tuple{\gamma_1',\tau_1'}\in\signalmap(\w_1')$ implying $\gamma_1'\in\BehE{}$. \\
\item Since $\ElMaxE{}$ is the strongest async. $l$-complete approx. of $\EE{}$ we have $\BehE{}\subseteq\BehElaMax{}$ and therefore $\gamma_1'\in\BehElaMax{}$.\\
\item Now pick $z_1'\in\Z^{\Nbn}$ s.t. $\AllQ{k\in\Nbn}{z_1'(k)=\gamma_1'\ll{\max(0,k-l),k-1}}$ and observe that the construction of $\gamma_1'$ implies $z_1'(\kt)=z'(k_2)$ and since $\gamma_1'\in\BehElaMax{}$ it follows from \REFlem{lem:CorrPastIndSS} by construction that $\Tuple{\z_1',\gamma_1'}\in\BehlaMaxS{}$.\\
\item Now pick $z_2:=\CONCAT{z'}{k_2}{\kt}{z_1'}$ and $\gamma_2:=\CONCAT{\gamma'}{k_2}{\kt}{\gamma_1'}$. Since $\ElaMaxS{}=\Tuple{\Nbn,\Gamma,\Z,\BehlaMaxS{}}$ is an asynchronous state space system, $\Tuple{\z_1,\gamma_1},\Tuple{\z_1',\gamma_1'}\in\BehlaMaxS{}$ and $z_1'(\kt)=z'(k_2)$ from above, \eqref{equ:StateSpaceDynamicalSystem:2} implies $\Tuple{\z_2,\gamma_2}\in\BehlaMaxS{}$.\smalllb
\end{inparaitem}
\item Finally we show, that for this choice of $\z_2$ and $\gamma_2$ the right side of the implication in \eqref{equ:def:SimRel:b} holds: \\
\begin{inparaitem}[$\blacktriangleright$\hspace{-0.15cm}]
\item Show $\gamma_2=\CONCAT{\gamma'}{k_2}{k_1}{\gamma_1}$: Observe that
\[\gamma_2=\CONCAT{\gamma'}{k_2}{\kt}{\gamma_1'}=\CONCAT{\gamma'}{k_2}{\kt}{\BR{\CONCAT{\gt}{\kt}{k_1}{\gamma_1}}}=\CONCAT{\gamma'}{k_2}{k_1}{\gamma_1}.\]
\item Show $\AllQ{k< k_2}{\propConj{\gamma_2(k)=\gamma'(k)}{\z_2(k)=z'(k)}}$: Follows from \eqref{equ:concat} and the construction of $\z_2$ and $\gamma_2$.\\
\item Show $z_2(k_2)=z'(k_2)$: From \eqref{equ:concat} we have $z_2(k_2)=z_1'(\kt)$ and from above $z_1'(\kt)=z'(k_2)$.\\
\item Show:
\[\AllQ{k\geq k_2,t_1'\inps\timescaleDown{\tau_1}(k-k_2+k_1),t_1'>t_1}{
 \Tuple{x_1(t_1'),z_2(k)}\in\R_0}\] 
\noindent With $\hat{k}:=k-k_2+\kt$ and $z_2=\CONCAT{z'}{k_2}{\kt}{z_1'}$ we get 
 \[\AllQ{\hat{k}\geq \kt,t_1'\in\timescaleDown{\tau_1}(\hat{k}-\kt+k_1),t_1'>t_1}{
 \Tuple{x_1(t_1'),z_1'(\hat{k})}\in\R_0}.\]
\noindent  Using $\tau_1'^{-1}=\CONCAT{\taut^{-1}}{\kt}{k_1}{(\tau_1^{-1}+\tilde{t}-t_1)}$, $x_1'=\CONCAT{\xt}{\tet}{t_1}{x_1}$ and $\hat{t}:=t_1'+\tet-t_1$ gives
 \[\AllQ{\hat{k}\geq \kt,\hat{t}\in\timescaleDown{\tau_1'}(\hat{k}),\hat{t}>\tet}{ \Tuple{x_1'(\hat{t}),z_1'(\hat{k})}\in\R_0}.\] 
 With the construction of $\R_0$ and $z_1'$ the statement to be proven is therefore true iff the statement
  \[\AllQ{\hat{k}\geq \kt,\hat{t}\in\timescaleDown{\tau_1'}(\hat{k}),\hat{t}>\tet}{ x_1'(\hat{t})\in
 \DiCases{\Xxr{}{\hat{k}}{\gamma_1'\ll{0,\hat{k}-1}}}{\hat{k}<l}{\Xx{}{\gamma_1'\ll{\hat{k}-l,\hat{k}-1}}}{\hat{k}\geq l}}\]
holds, what is true from \REFdef{def:TimeIndexedStateSpace1}, since $(w_1',x_1')\in\BehS{}$ and $\Tuple{\gamma_1',\tau_1'}\in\signalmap(\w_1')$ from above
\end{inparaitem}.\hspace*{0pt}\hfill\begin{small}$\blacksquare$\end{small}
\end{compactitem}

\vspace{0.1cm}

\textit{Proof of \REFlem{lem:ElmaxSpiSimulatesEs} (ii):}\\
\textbf{1)}
\textbf{Show \eqref{equ:def:InitSimRel} holds for $l$:}
\begin{compactitem}
 \item Fix any $\Tuple{w,x}\in\BehS{},\Tuple{\gamma,\tau}\in\signalmap(w),t\in\timescaleDown{\tau}(l)$. 
 \item Observe that $\gamma\in\BehE{}$ (from \REFdef{def:DynSysInducesFiniteV}), $\gamma\in\BehElaMax{}$ (from \REFlem{lem:constructElMax_general_TA}) and $x(t)\in\Xk{l}{}$ (from \REFdef{def:TimeIndexedStateSpace1}). 
 \item This implies that we can pick $\zeta=\gamma\ll{0,l-1}$ and have $\zeta\in\Zk{l}{}$ (from \REFlem{lem:CorrPastIndSS}) and $x(t)\in\Xx{}{\zeta}$ (from \eqref{equ:SetOfReachableStates}), 
i.e., \eqref{equ:def:InitSimRel} holds. 
\end{compactitem}
\textbf{2)} \textbf{Show \eqref{equ:def:SimRel:b} holds}: This proof is identical to (i.2) only considering states $\zeta$ of length $l$, and is therefore omitted.%
\hspace*{10pt}\hfill\begin{small}$\blacksquare$\end{small}
 \vspace{0.08cm}

\textit{Proof of \REFlem{lem:EsLSimulatesElmaxSpi} \enquote{$\Rightarrow$}:}\\
\textbf{1.) Show \eqref{equ:def:InitSimRel} holds for $\R_l^{-1}$:}
 \begin{compactitem}
 \item Pick $\zeta\in\Zk{l}{}$ and observe, that from \REFlem{lem:CorrPastIndSS} and $\BehE{}=\BehlMaxE{}$ follows that there exits $\Tuple{w,x}\in\BehS{},\Tuple{\gamma,\tau}\in\signalmap(w),t\in\timescaleDown{\tau}(l)$ s.t. $\gamma\ll{0,l-1}=\zeta$ implying $x(t)\in\Xk{l}{}$ and $x(t)\in\Xx{}{\zeta}$, i.e., \eqref{equ:def:InitSimRel} holds.
 \end{compactitem}
\textbf{2.) Show \eqref{equ:def:SimRel:b} holds for $\R_l^{-1}$}: 
\begin{compactitem}
 \item Using the trivial signal map $\philaMax$ we fix $\Tuple{\gamma_1,z_1}\in\BehlaMaxS{},\Tuple{w',x'}\in\BehS{}, \Tuple{\gamma',\tau'}\in\signalmap(w'), t_2\in\T, k_1,k_2\in\Nbn$ s.t. the left side of the implication in \eqref{equ:def:SimRel:b} holds, i.e., $k_2=\timescaleUp{\tau_2}(t_2)$ and $\Tuple{z_1(k_1),x'(t_2)}\in\R_l^{-1}$.
\item Now we construct particular $x_2,w_2,\tau_2$ and $\gamma_2$ to show that the right side of the implication in \eqref{equ:def:SimRel:b} holds:\\
\begin{inparaitem}[$\blacktriangleright$\hspace{-0.15cm}]
 \item Since $\Tuple{\gamma_1,z_1}\in\BehlaMaxS{}$ we have $\gamma_1\in\BehElaMax{}$ from \REFlem{lem:CorrPastIndSS} and as $\BehE{}$ is $l$-complete, we have $\BehE{}=\BehElaMax{}$ and therefore $\gamma_1\in\BehE{}$. Using \eqref{equ:BehE}, we can therefore fix $\Tuple{w_1,x_1}\in\BehS{},\tau_1\in\timescale$ s.t. $\Tuple{\gamma_1,\tau_1}\in\signalmap(w_1)$. \\
 \item Remember $\Tuple{\gamma_1,z_1}\in\BehlaMaxS{}$ and $\Tuple{z_1(k_1),x'(t_2)}\in\R_l^{-1}$. Using \eqref{equ:SetOfReachableStates}, we can pick $t_1\in\timescaleDown{\tau_1}(k_1)$ and have $x_1(t_1)\in\Xx{}{z_1(k_1)}$. Using \eqref{equ:Rx} this implies $\Tuple{x_1(t_1),x'(t_2)\in\R_{\mathcal{X}}}$.\\
\item Since $\R_{\mathcal{X}}\in\SR{l}{\EpS{}}{\EpS{}}$ we know that we can pick $\Tuple{w_2,x_2}\in\BehS{},\Tuple{\gamma_2,\tau_2}\in\signalmap(w_2)$ s.t. 
\begin{equation}\label{equ:proof:1}
  \begin{propConjA}
\gamma_2=\CONCAT{\gamma'}{k_2}{k_1}{\gamma_1}\\
\AllQ{t\in\T_2, t< t_2}{
\begin{propConjA}
w_2(t)=w'(t)\\
x_2(t)=x'(t)\\
\tau_2(t)=\tau'(t)\\
\end{propConjA}
}\\
x_2(t_2)=x'(t_2)\\
\AllQSplit{k\geq k_2,t_1'\in\timescaleDown{\tau_1}(k-k_2+k_1),t_1'>t_1}{\ExQ{t_2'\in\timescaleDown{\tau_2}(k),t_2'>t_2 }{
\Tuple{x_1(t_1'),x_2(t_2')}\in\R_{\mathcal{X}}}}
\end{propConjA}
.
\end{equation}
\end{inparaitem}
\item Observe, that with this choice of $w_2,x_2,\gamma_2,\tau_2$ the first tree lines of the right side of \eqref{equ:def:SimRel:b} for $\R_l^{-1}$ are equivalent to the first tree lines in \eqref{equ:proof:1}.
\item Now we show that the last line of \eqref{equ:def:SimRel:b} holds for $\R_l^{-1}$.\\
\begin{inparaitem}[$\blacktriangleright$\hspace{-0.15cm}]
\item Observe ${x_2(t_2)=x'(t_2)\in\Xx{}{z(k_1)}}$. Using \eqref{equ:SetOfReachableStates} we can therefore fix $\Tuple{\wt,\xt}\in\BehS{},\Tuple{\gt,\taut}\in\signalmap(\wt),\kt\in\Nbn,\tet\in\timescaleDown{\tau}(\kt)$
 s.t. $\x_2(t_2)=\xt(\tet)$ and
 $z(k_1)=\gt\ll{\kt-l,\kt-1}$, implying $\gt\ll{\kt-l,\kt-1}=\gamma_1\ll{k_1-l,k_1-1}$. \\
\item Since $\EpS{}$ is an asynchronous state space dynamical system, we can pick $w''=\CONCAT{\wt}{\tet}{t_2}{w_2}$,$x''=\CONCAT{\xt}{\tet}{t_2}{x_2}$,$\tau''=\CONCAT{\taut}{\tet}{t_2}{(\tau_2+\tilde{k}-k_2)}$ and $\gamma''=\CONCAT{\gt}{\kt}{k_2}{\gamma_2}$ and have $(w'',x'')\in\BehS{}$ and $\Tuple{\gamma'',\tau''}\in\signalmap(\w'')$.\\
\item Now using $\gamma_2=\CONCAT{\gamma'}{k_2}{k_1}{\gamma_1}$ and $\gt\ll{\kt-l,\kt-1}=\gamma_1\ll{k_1-l,k_1-1}$ from above gives 
$\gamma''=\CONCAT{\gt}{\kt}{k_2}{\gamma_2}=
\CONCAT{\gt}{\kt}{k_2}{\BR{\CONCAT{\gamma'}{k_2}{k_1}{\gamma_1}}}
=\CONCAT{\gt}{\kt}{k_1}{\gamma_1}
=\CONCAT{\gt}{\kt-l}{k_1-l}{\gamma_1}
$
implying 
\begin{equation}\label{equ:proof:lem:EsLSimulatesElmaxSpi:1}
 \AllQ{\hat{k}\geq \kt}{\gamma''\ll{\hat{k}-l,\hat{k}-1}=\gamma_1\ll{\hat{k}-\kt+k_1-l,\hat{k}-\kt+k_1-1}}.
\end{equation}
\item Remember that we have to show
 \[\AllQ{k\geq k_2}{\ExQ{t_2'\in\timescaleDown{\tau_2}(k),t_2'>t_2 }{
\Tuple{z_1(k-k_2+k_1),x_2(t_2')}\in\R_l^{-1}.}}\]
Using $\hat{k}:=k-k_2+\kt$, $\R_l^{-1}$ from \eqref{equ:lem:ElmaxSpiSimulatesEs:Rl} and  $(\gamma_1,z_1)\in\BehlaMaxS{}$ this is equivalent to
\[\AllQ{\hat{k}\geq \kt}{\ExQ{t_2'\in\timescaleDown{\tau_2}(\hat{k}-\kt+k_2),t_2'>t_2 }{
x_2(t_2')\in\Xx{}{\gamma_1\ll{\hat{k}-\kt+k_1-l,\hat{k}-\kt+k_1-1}},}}\]
and using $\tau''^{-1}=\CONCAT{\taut^{-1}}{\kt}{k_2}{(\tau_2^{-1}+\tilde{t}-t_2)}$, $x''=\CONCAT{\xt}{\tet}{t_2}{x_2}$, \eqref{equ:proof:lem:EsLSimulatesElmaxSpi:1} and $\hat{t}:=t_2'+\tet-t_2$ it is equivalent to
 \[\AllQ{\hat{k}\geq \kt}{\ExQ{\hat{t}\in\timescaleDown{\tau''}(\hat{k}),\hat{t}>\tet }{
x''(\hat{t})\in\Xx{}{\gamma''\ll{\hat{k}-l,\hat{k}-1}}}}.\]
Now observe that the last statement is true from \eqref{equ:SetOfReachableStates}, since $(w'',x'')\in\BehS{}$ and $\Tuple{\gamma'',\tau''}\in\signalmap(\w'')$ from above, what proves the statement
\end{inparaitem}.\hspace*{0pt}\hfill\begin{small}$\blacksquare$\end{small}
\end{compactitem}

\vspace{0.1cm}
%
\textit{Proof of \REFlem{lem:EsLSimulatesElmaxSpi} \enquote{$\Leftarrow$}:}\\
\textbf{1.) Show ${\BehE{}=\BehlMaxE{}}$:} 
\begin{compactitem}
\item Observe that $\BehE{}\subseteq\BehlMaxE{}$ by definition and $\R_l^{-1}\in\SR{l}{\EplaMaxS{}}{\EpS{}}$ implies $\BehE{}\llr{l,\infty}=\BehlMaxE{}\llr{l,\infty}$.
\item Show $\BehlMaxE{}\ll{0,l-1}\subseteq\BehE{}\ll{0,l-1}$: \\
\begin{inparaitem}[$\blacktriangleright$\hspace{-0.15cm}]
 \item Observe that \REFlem{lem:CorrPastIndSS} implies $\BehlMaxE{}\ll{0,l-1}=\Zk{l}{}$. \\
 \item Now \eqref{equ:lem:ElmaxSpiSimulatesEs:Rl} and $\R_l^{-1}\in\SR{l}{\EplaMaxS{}}{\EpS{}}$ implies $\AllQ{\zeta\in\Zk{l}{}}{\ExQ{\xi\in\Xk{l}{}}{\xi\in\Xx{}{\zeta}}}$.\\ 
 \item With \REFdef{def:TimeIndexedStateSpace1}, this implies $\AllQ{\zeta\in\Zk{l}{}}{\ExQ{\gamma\in\BehE{}}{\gamma\ll{0,l-1}=\zeta}}$ what proves the statement.
\end{inparaitem}
\end{compactitem}
\textbf{2.) Show \eqref{equ:def:InitSimRel} holds for $\R_{\mathcal{X}}$:} by construction.\\ 
\textbf{3.) Show \eqref{equ:def:SimRel:b} holds for $\R_{\mathcal{X}}$}: 
\begin{compactitem}
 \item Pick the signals $\cdot_1$ and $\cdot'$ as in the first part, i.e. \enquote{$\Rightarrow$}, of the proof and observe, that $\Tuple{x_1(t_1),x'(t_2)}\in\R_{\mathcal{X}}$ and $\Tuple{z_1(k_1),x'(t_2)}\in\R_l^{-1}$ as before.
 \item Since $\R_l^{-1}\in\SR{l}{\EplaMaxS{}}{\EpS{}}$ we can use \eqref{equ:def:SimRel:b} and pick signals $\Tuple{w_2,x_2}\in\BehS{},\Tuple{\gamma_2,\tau_2}\in\signalmap(w_2)$ s.t.
 \begin{equation}\label{equ:proof:2}
  \begin{propConjA}
\gamma_2=\CONCAT{\gamma'}{k_2}{k_1}{\gamma_1}\\
\AllQ{t\in\T_2, t< t_2}{
\begin{propConjA}
w_2(t)=w'(t)\\
x_2(t)=x'(t)\\
\tau_2(t)=\tau'(t)\\
\end{propConjA}
}\\
x_2(t_2)=x'(t_2)\\
\AllQ{k\geq k_2}{\ExQSplit{t_2'\in\timescaleDown{\tau_2}(k),t_2'>t_2 }{
\Tuple{z_1(k-k_2+k_1),x_2(t_2')}\in\R_l^{-1}}}
\end{propConjA}
.
\end{equation}
\item Again, the first tree lines of \eqref{equ:proof:2} are equivalent to \eqref{equ:proof:1}.
\item To proof, that the last line of  \eqref{equ:proof:1} also holds for this choice of signals $\cdot_2$,
remember that by construction $\AllQ{k\in\Nbn,t\in\timescaleDown{\tau_1}(k)}{x_1(t)\in\Xx{}{z_1(k)}}$. With the parametrization of $k$ in the last line of \eqref{equ:proof:1} and the last line of \eqref{equ:proof:2} we therefore have 
\[\AllQSplit{k\geq k_2,t_1'\in\timescaleDown{\tau_1}(k-k_2+k_1),t_1'>t_1}{\ExQ{t_2'\in\timescaleDown{\tau_2}(k),t_2'>t_2 }{
\Tuple{x_1(t_1'),x_2(t_2')}\in\Xx{}{z_1(k-k_2+k_1)}}}\]
what proves the statement (from \eqref{equ:Rx}).\hspace*{0pt}\hfill\begin{small}$\blacksquare$\end{small}%
\end{compactitem}%

\end{appendix}

\end{document}